%% file: 00_main.tex
\documentclass[conference]{IEEEtran}
\IEEEoverridecommandlockouts

\usepackage{microtype}
\usepackage{amssymb}
\usepackage{xspace}	
\usepackage{amsfonts,amsmath,amsthm}
\usepackage{graphicx}
\usepackage{hyperref}
\usepackage[utf8]{inputenc}
\usepackage{xcolor}
\usepackage{enumitem}
\usepackage[nothing]{algorithm}
\usepackage{algorithmicx}
\usepackage[noend]{algpseudocode}

\input{macros.tex}

\graphicspath{{./FIGURES/}}

\usepackage{wrapfig}

\usepackage{hyperref}
\hypersetup{
    linktoc=all,
    colorlinks=false,        
    linkcolor=black,         
    linkbordercolor=black,
    citecolor=black,         
    filecolor=black,         
    urlcolor=black,          
    pdfborderstyle={/S/U/W 0}
}

\title{Cross-chain Swaps with Preferences
\thanks{This project has been supported by NSF Career 1942711, DARPA YFA D22AP00146, and NSF grant CCF-2153723.}
}

\author{\IEEEauthorblockN{Eric Chan}
\IEEEauthorblockA{University of California at Riverside}
\and
\IEEEauthorblockN{Marek Chrobak}
\IEEEauthorblockA{University of California at Riverside}
\and
\IEEEauthorblockN{Mohsen Lesani}
\IEEEauthorblockA{University of California at Riverside}
}

\begin{document}


\maketitle

\begin{abstract}
Extreme valuation and volatility of cryptocurrencies require investors to diversify often which demands secure exchange protocols.
A cross-chain swap protocol allows distrusting parties to securely exchange their assets.
However, the current models and protocols assume predefined user preferences for acceptable outcomes.
This paper presents a generalized model of swaps that allows each party to specify its preferences on the subsets of its incoming and outgoing assets.
It shows that the existing swap protocols are not necessarily a strong Nash equilibrium in this model.
It characterizes the class of swap graphs that have protocols that are safe, live and a strong Nash equilibrium, and presents such a protocol for this class.
Further, it shows that deciding whether a swap is in this class is NP-hard through a reduction from 3SAT,
and further is 
$\SigmaTwoP$-complete
through a reduction from
$\ExistsForallDNF$.
\end{abstract}  


\section{Introduction}
\label{sec:introduction}
\input{01_Introduction.tex}


\section{Swap Systems}
\label{sec:swap systems}
\input{02_swap_systems.tex}


\section{Herlihy's Swap Model}
\label{sec:h-swap systems}
\input{03_h-swap_systems.tex}


\section{A Characterization of Swap Systems with Atomic Protocols}
\label{sec:uniandnash}
\input{04_characterization.tex}


\section{NP-Hardness}
\label{sec:np_hardness_results}
\input{05_np_hardness.tex}


\section{$\Sigma_2$-Completeness}
\label{sec:sigma2 completeness-main}
\input{06_sigma_two_main.tex}


\section{Related Works}
\label{sec:related}
\input{07_related.tex}

\section{Conclusion}
We presented a general swap model that allows each party to specify their preference on their possible outcomes.
We saw that Herlihy's pioneering protocol is a uniform and Nash strategy in this model; however, it is not a strong Nash strategy.
We presented a characterization of the class of swap graphs that have uniform and Strong Nash protocols.
Interestingly, Herlihy's protocol is such a strategy when executed on a particular subgraph of the swap graphs in this class.
We further presented reductions that shows the NP-harness and $\SigmaTwoP$-completeness of the decision problem for this class.

\myparagraph{Acknowledgements}
I would like to thank Annie Semb for her unconditional love and support over the years. 
I will remember you always.
I miss you dearly.
\emph{Eric.}


\bibliographystyle{plainurl}
\bibliography{references}


\clearpage
\noindent
{\Huge Appendix}

\section{$\Sigma_2$-Completeness}
\label{sec: sigma2 completeness}
\input{11_sigma2_completeness.tex}

\section{Another Proof of $\NP$-Hardness}
\label{sec: simple np-hardness}
\input{12_simplest_np_hardness.tex}

\section{Experiments}
\label{sec: experiments}
\input{13_supplemental_experiments.tex}

\end{document}

%% file: macros.tex
\newcommand{\braced}[1]{{ \left\{ #1 \right\} }}
\newcommand{\smbraced}[1]{{\{ #1 \} }}
\newcommand{\angled}[1]{{ \langle #1 \rangle }}

\newcommand{\barred}[1]{{ \left| #1 \right| }}

\newcommand{\suchthat}{{\;:\;}}

\newcommand{\barx}{{\bar x}}
\newcommand{\bary}{{\bar y}}
\newcommand{\barz}{{\bar z}}

\newcommand{\barOmega}{{\bar\Omega}}

\newcommand{\bfx}{{\bf x}}
\newcommand{\bfy}{{\bf y}}

\newcommand{\bfphi}{{\boldsymbol{\phi}}}
\newcommand{\bfpsi}{{\boldsymbol{\psi}}}
\newcommand{\bfxi}{{\boldsymbol{\xi}}}

\newcommand{\calC}{\mathcal{C}}

\newcommand{\ttA}{{\texttt{A}}}
\newcommand{\ttB}{{\texttt{B}}}

\newcommand{\ttX}{{\texttt{X}}}
\newcommand{\ttY}{{\texttt{Y}}}

\newcommand{\tildex}{{\tilde{x}}}
\newcommand{\tildey}{{\tilde{y}}}
\newcommand{\tildez}{{\tilde{z}}}

\newcommand{\bartildex}{{\bar{\tilde{x}}}}
\newcommand{\bartildey}{{\bar{\tilde{y}}}}
\newcommand{\bartildez}{{\bar{\tilde{z}}}}

\newcommand{\myparagraph}[1]{{\medskip\noindent\textbf{#1}.}}

\newcommand{\myclaim}[2]{{\smallskip\noindent\textit{Claim~$#1$}: \textit{#2}}}

\newcommand{\PP}{{\mathsf{P}}}
\newcommand{\NP}{{\mathsf{NP}}}

\newcommand{\SigmaTwoP}{{\Sigma_2^{{\PP}}}}

\newcommand{\swapAtomic}{\mathsf{SwapAtomic}\xspace}

\newcommand{\CNF}{\mathsf{CNF}\xspace}
\newcommand{\ExistsForallDNF}{\exists\forall\mathsf{DNF}\xspace}
\newcommand{\ExistsForallThreeDNF}{\exists\forall\mathsf{3DNF}\xspace}
\newcommand{\ExistsForallDNFOneX}{\exists\forall\mathsf{DNF}_{1x}\xspace}

\newcommand{\Deal}{\textsc{Deal}\xspace}
\newcommand{\Dealv}[1]{{\textsc{Deal}_{#1}}\xspace}
\newcommand{\DealDv}[2]{{\textsc{Deal}^{#1}_{#2}}\xspace}
\newcommand{\NoDeal}{\textsc{NoDeal}\xspace}
\newcommand{\NoDealv}[1]{{\textsc{NoDeal}_{#1}}\xspace}
\newcommand{\NoDealDv}[2]{{\textsc{NoDeal}^{#1}_{#2}}\xspace}
\newcommand{\Discount}{\textsc{Discount}\xspace}
\newcommand{\Discountv}[1]{{\textsc{Discount}_{#1}}\xspace}
\newcommand{\DiscountDv}[2]{{\textsc{Discount}^{#1}_{#2}}\xspace}
\newcommand{\FreeRide}{\textsc{FreeRide}\xspace}
\newcommand{\FreeRidev}[1]{{\textsc{FreeRide}_{#1}}\xspace}
\newcommand{\FreeRideDv}[2]{{\textsc{FreeRide}^{#1}_{#2}}\xspace}
\newcommand{\Underwater}{\textsc{Underwater}\xspace}
\newcommand{\Underwaterv}[1]{{\textsc{Underwater}_{#1}}\xspace}

\newcommand{\outcome}{{\omega}}
\newcommand{\outcomepair}[2]{\angled{#1\,|\,#2}}
\newcommand{\vecoutcome}{{\bar{\omega}}}
\newcommand{\outcomein}[1]{{\outcome_{#1}^{in}}}
\newcommand{\outcomeout}[1]{{\outcome_{#1}^{out}}}
\newcommand{\houtcome}{{\hat{\omega}}}
\newcommand{\houtcomein}[1]{{\houtcome_{#1}^{in}}}
\newcommand{\houtcomeout}[1]{{\houtcome_{#1}^{out}}}

\newcommand{\Ain}[1]{{A_{#1}^{in}}}

\newcommand{\Aout}[1]{{A_{#1}^{out}}}

\newcommand{\Ayes}[1]{\mathcal{A}_{#1}}

\newcommand{\swapSys}{{\mathcal{S}}}

\newcommand{\prefP}{{\mathcal{P}}}
\newcommand{\acceptsetA}{{\mathcal{A}}}
\newcommand{\DG}{{\mathcal{D}}}
\newcommand{\GG}{{\mathcal{G}}}
\newcommand{\HG}{{\mathcal{H}}}
\newcommand{\coalition}{{\mathcal{C}}}

\newcommand{\bbP}{{\mathbb{P}}}
\newcommand{\hProt}{{\mathbb{H}}}

\newtheorem{theorem}{Theorem}
\newtheorem{lemma}{Lemma}

\renewcommand{\paragraph}[1]{\textit{#1. \ }}

%% file: 01_Introduction.tex

The multitude and volatility of cryptocurrencies force investors to diversify and frequenty trade their holdings.
However, these currencies are hosted by distinct distributed blockchains and  trading across blockchains is not atomic by default.
This has led to the development of \emph{cross-chain swap protocols}
\cite{Herlihy18,btcwiki,bitcoinhtlc,heilman2020arwen,decred,thyagarajan2022universal} 
that allow distrusting parties to securely exchange their assets.
Application of such swap protocols is not limited to trading digital currencies --- they can be 
used for trading any type of digital assets (NFTs, for example), or even for trading
physical objects by safely trasferring ownership documentation.

In a pioneering work, Herlihy \cite{Herlihy18}
formalizes a cross-chain swap as a directed graph where vertices represent parties, 
and arcs represent assets to be exchanged.
An execution of a \emph{swap graph} is represented as the subset of arcs that are triggered in that execution.
The \emph{outcome} for each party is captured in five \emph{predefined classes}:
$\Deal$, $\NoDeal$, $\Discount$, $\FreeRide$, and $\Underwater$.
The classes $\Deal$ and $\NoDeal$ represent outcomes for a party where respectively, all and none of
the arcs of that party are triggered.
The class $\Discount$ represents outcomes where some of the outgoing arcs are not triggered, and 
$\FreeRide$ represents outcomes where at least one incoming but no outgoing arc is triggered.
Outcomes in all these classes are considered acceptable by each party.
The class $\Underwater$ captures all unacceptable outcomes, namely
outcomes where at least one outgoing arc is triggered but not all incoming arcs are.
Given this model of outcomes, Herlihy presented a \emph{protocol based on hashed time-locks}
and proved it to be \emph{atomic}, meaning that it satisfies the conditions of
\emph{liveness, safety and strong Nash equilibrium}.


In practice, as noted in the original proposal \cite{Herlihy18}, some parties
may find it advantageous to exchange only \emph{some} of their outgoing 
assets for only \emph{some} of their incoming assets.  As an example, suppose that
Alina has a white shirt and white pants and she joins the swap hoping to trade for a black shirt and black pants.
Coincidentally, Bohdan has exactly these items and joins the swap looking for the reverse trade.
However, both of them would actually prefer to retain one white article of clothing and one black article of clothing, if possible.
Thus, it would be preferable for both parties to, say, only swap the shirts or only swap the pants,
although it is also acceptable to swap both. Such scenarios are
not captured by the model in \cite{Herlihy18}, because the outcomes with just one item swapped
are in the class $\Underwater$.

This leads to the natural question, left open in \cite{Herlihy18}:
is there a more general swap model that allows each party to specify its personal preferences over all possible
swap outcomes, and, at the same time, admits an atomic protocol.

Addressing this question, this paper introduces a general model of cross-chain swaps that we call \emph{swap systems}. 
In a swap system, as in \cite{Herlihy18}, the set of prearranged asset transfers is represented
by a directed graph. Unlike in  \cite{Herlihy18}, however, in our model each party 
can specify its own preferences between all its possible outcomes (that is, between sets consisting of its incoming and
outgoing arcs). These preferences can be arbitrary, as long as they form a poset and satisfy natural monotonicity conditions.
This generality allows us to capture not only subjectivity of preferences, but also dependencies between assets. 
The example above (about trading clothing items) illustrates such a dependency: 
for the purpose of trading, Alina values her pair of items higher than the sum of their individual values.
Such dependencies often arise in practice when a party intends to trade multiple assets --- in fact, 
common investment strategies are guided by objectives (diversification, for example) that
inherently involve asset dependencies.

As it turns out, Herlihy's protocol is not necessarily atomic in all swap systems, although
it still satisfies the conditions of liveness, safety, and weak Nash equilibrium. 
We then present a characterization of swap systems that admit atomic protocols.
The correctness proof of this characterization embodies such a protocol.
We then focus on the problem of verifying whether a given swap systems has an atomic protocol.
To this end, we provide a full characterization of the time complexity of this problem and show that it's computationally infeasible,
by a novel proof of completeness in the complexity class $\SigmaTwoP$. 
As a stepping stone to this full characterization, 
we also include an easier proof of NP-hardness.

The paper is organized as follows.
\begin{itemize}
\item 
In~\autoref{sec:swap systems} we introduce our model of swap systems,
including the definitions of atomic protocols in this model.

\item
In~\autoref{sec:h-swap systems} we show that our model is indeed a generalization 
of Herlihy's model.

\item 
The full characterization of swap systems that admit atomic protocols is
given in~\autoref{sec:uniandnash}.

\item
The decision problem of testing whether a swap systems admits an atomic
protocol is studied in~\autoref{sec:np_hardness_results}
and \autoref{sec:sigma2 completeness-main}, first proving
$\NP$-hardness and then refining the proof to show $\SigmaTwoP$-completeness.

\end{itemize}

For readers interested in the practical impact of our work, the overall take-out message from this paper is this:
(i) 
Even if some parties wish to specify outcome preferences not captured by the model in \cite{Herlihy18},
it still may be possible to realize the swap with a protocol that is atomic and efficient.
(ii) 
The challenge is that in order to determine whether it is possible, and
to actually specify this protocol, one needs to solve a computationally infeasible 
decision problem.
Naturally, for small number of parties this can still be done in practice 
-- say by exhaustive search.


%% file: 02_swap_systems.tex

As discussed in the introduction, Herlihy's model~\cite{Herlihy18} for cross-chain swaps assumed that the rational
behavior of participating parties is determined by preferences between five types of outcomes:
$\Deal$, $\NoDeal$, $\Discount$, $\FreeRide$, and $\Underwater$. 
These preferences were assumed to be shared by all parties, and can be interpreted as a simple partial order on all possible outcomes.
Some of these preferences are natural; for example, in $\Discount$ a party receives all incoming assets without
trading all outgoing assets, making it preferable to $\Deal$. But, as explained in the introduction, in practice a party
may consider some outcomes designated as $\Underwater$ in~\cite{Herlihy18} to be acceptable, or even preferable to $\Deal$. 
As another example, suppose that Alina possesses items $\ttA$ and $\ttB$ that she values at $\$10$ and $\$12$,
and Bohdan possesses items $\ttX$ and $\ttY$ that Alina values at $\$11$ and $\$14$.
Providing that Alina's preferences are based only on the monetary value,
she would accept to join the swap that allows her to swap both $\ttA$ and $\ttB$ for Bohdan's $\ttX$ and $\ttY$, 
but she would be even happier if she ends up swapping only $\ttA$ for $\ttY$ instead.
Similarly, there is no justification for the outcomes in $\FreeRide$ to be incomparable to $\Deal$ or $\Discount$. 

To represent such individual preferences, 
we now refine Herlihy's model by allowing each party to specify a partial order on all her possible outcomes
of a protocol. Our model is very general in that (unlike in the example above) a party's preferences are
not determined by numerical values of individual assets, but rather involve comparing directly whole sets of 
traded and acquired assets. The advantage of this approach is that it captures dependencies between assets, 
when a party values a set of assets higher or lower than the sum of their individual values. As an example,
say that Alina owns a power drill and a shovel, while Bohdan is in possession of a pair of skis.
Alina would not swap any of her items for any single ski, but she may be happy to swap both of her items for the pair.
On the other hand, if, instead of skis, Bohdan needs to get rid of two skateboards, Alina may prefer to 
swap any of her items for one skateboard rather than swapping both for two skateboards.


\myparagraph{Swap Systems}
A \emph{swap system} is specified by a pair $\swapSys = (\DG, \prefP)$ consisting of a digraph $\DG$
that represents the pre-arranged asset transfers and a collection $\prefP$ of posets that specifies 
the preferences of each involved party among all of its potential outcomes. Next, we give a formal definition
of these two components of $\swapSys$.

Digraph $\DG = (V, A)$ is called a \emph{swap digraph}. Each vertex $v\in V$ represents a party that participates in the swap, and
each arc $(u,v) \in A$ represents an asset that is to be transferred from party $u$ to party $v$. 
By $\Ain{v}$ and $\Aout{v}$ we will denote the sets of vertex $v$'s incoming and outgoing arcs, respectively.
If $(x,v) \in \Ain{v}$ then $x$ is called an \emph{in-neighbor} of $v$,
and if $(v,x) \in \Aout{v}$ then $x$ is called an \emph{out-neighbor} of $v$.
Throughout the paper we assume that $\DG$ does not have multiple arcs\footnote{%
This assumption is only for convenience -- our model and results trivially extend to multi-digraphs,
although this requires more cumbersome notation and terminology.
}.
We also assume that $\DG$ is weakly
connected (otherwise a swap can be arranged for each connected component separately). 
To exclude some degenerate scenarios, we also assume that $|V|\ge 2$ and that 
$\Ain{v}\neq\emptyset$ and $\Aout{v}\neq\emptyset$ for each $v\in V$.

An \emph{outcome} of a party $v\in V$ is a pair $\outcome = \outcomepair{\outcomein{}}{\outcomeout{}}$, where
$\outcomein{} \subseteq \Ain{v}$ and $\outcomeout{} \subseteq \Aout{v}$. An outcome represents the sets of 
acquired and traded assets, $\outcomein{}$ and $\outcomeout{}$ respectively.
The set of all possible outcomes of $v$ will be denoted $\Omega_v$.
To reduce clutter, instead of arcs, in $\outcomepair{\outcomein{}}{\outcomeout{}}$ we
will often list only the corresponding in-neighbors and out-neighbors of $v$; for example, instead of 
$\outcomepair{\braced{(x,v),(y,v)}}{\braced{v,z}}$ we will write $\outcomepair{x,y}{z}$.

The collection $\prefP = \braced{\prefP_v}_{v\in V}$ consists of \emph{preference posets}.
The preference poset of a party $v\in V$ is $\prefP_v = (\Omega_v, \preceq_v)$, where $\preceq_v$ is a partial order on $\Omega_v$.
We will write $\outcome \prec_v \outcome'$ if $\outcome \preceq_v \outcome'$ and $\outcome \neq \outcome'$.
This poset naturally represents $v$'s evaluation of its potential outcomes; that is, relation $\outcome \preceq_v \outcome'$ holds
if $v$ views outcome $\outcome'$ to be better than outcome $\outcome$.
The outcome where $v$ does not participate in any transfer is  $\NoDealDv{}{v} = \outcomepair{\varnothing}{\varnothing}$
and the outcome where all of $v$'s transfers are realized is  $\DealDv{}{v} = \outcomepair{\Ain{v}}{\Aout{v}}$. 
Each preference poset $\prefP_v$ is assumed to have the following properties:
\begin{description}[leftmargin=0.1in]
    \item{(p.1)} $\Deal$ is better than $\NoDeal$:  $\NoDealDv{}{v} \prec_v \DealDv{}{v}$.
    Naturally, each party prefers swapping all assets over being completely excluded, as otherwise it would not even
	join the swap system.
    \item{(p.2)} Inclusive Monotonicity:
    $(\outcomein{1} \subseteq \outcomein{2}\wedge \outcomeout{2} \subseteq \outcomeout{1})
	\Rightarrow \outcome_1 \preceq_v \outcome_2$, for every two outcomes $\outcome_1, \outcome_2 \in \Omega_v$.
	That is, it's better to receive more assets and to trade fewer assets\footnote{Duuh.}.
\end{description}
The preference pairs $\outcome_1 \prec_v \outcome_2$ that are determined by rules (p.1) and (p.2) above will be called
\emph{generic}. The size of the preference poset may be exponentially large with respect to the
size of the swap digraph $\DG$, but it is not necessary for a party to specify generic
preferences as they are implied from the above rules. Therefore,
throughout the paper, we assume that $\prefP_v$ is specified by its \emph{generator set}, which is a subset of its non-generic preference pairs that, 
together with the generic pairs and transitivity, generate the whole poset.
A generator set of a poset may not be unique. We use this convention in our examples and running time
bounds. (This does not affect our hardness results --- they hold even if the
preference poset of each party is specified by listing \emph{all} preference pairs.)

An outcome $\omega\in \Omega_v$ is called \emph{acceptable}  if $\omega\succeq \NoDealv{v}$.
The set of acceptable outcomes of a node $v$ will be denoted $\acceptsetA_{v}$\footnote{%
This definition can be relaxed to allow some outcomes incomparable to $\NoDeal$ be acceptable. In this extended model, the set $\acceptsetA_{v}$
of acceptable outcomes would be part of a swap system specification, and would have to satisfy three conditions:
(i) $\braced{\outcome\suchthat \outcome \succeq \NoDealv{v}} \subseteq \acceptsetA_{v}$,
(ii) $\braced{\outcome\suchthat \outcome \prec \NoDealv{v}} \cap \acceptsetA_{v} = \emptyset$, and
(iii) $\outcome\in \acceptsetA_{v} \wedge \outcome\preceq\outcome' \Rightarrow \outcome'\in \acceptsetA_{v}$.
Our results can be extended naturally to this model. We adopted the simpler definition to streamline the presentation.
}.

Throughout the paper, we will often omit subscript $v$ in these notations
(and others as well) if $v$ is implicit in the context or irrelevant.
On the other hand, if any ambiguity may arise, we will sometimes add a superscript to some notations
specifying the digraph under consideration; for example we
will write $\DealDv{\DG}{v}$ to specify that outcome $\DealDv{\DG}{v}$ is with respect to digraph $\DG$.


\myparagraph{Protocols}
Given a swap system $\swapSys = (\DG, \prefP)$, a \emph{swap protocol} $\bbP$ for $\swapSys$
specifies actions of each party over time, in particular it determines how assets change
hands. Initially, an asset represented by an arc $(u,v)\in A$ is in the possession of $u$, and,
when $\bbP$ completes, this asset must be in possession of either $u$ or $v$.
If $(u,v)$ ends up in the possession of $v$, we will say that the arc $(u,v)$ has been \emph{triggered}.
The outcome of $v$ after executing $\bbP$ is  $\outcomepair{\outcomein{}}{\outcomeout{}}$, where $\outcomein{}$ and $\outcomeout{}$
are the sets of incoming and outgoing arcs of $v$ that are triggered in this execution. In particular,
we write $\bbP(v)$ for the outcome of $v$ in an execution of protocol $\bbP$ in which all parties follow $\bbP$.
If some party (possibly $v$ itself) deviates from $\bbP$, we assume that $v$'s outcome is also finalized when $\bbP$ completes, 
but it may be different from $\bbP(v)$.

A protocol may use appropriate cryptographic primitives. In particular, following~\cite{Herlihy18},
we assume the availability of \emph{smart contracts}.
A smart contract for an arc $a = (u,v)$ allows $u$ to put asset $a$ in an escrow secured with a suitable collection of
hashed time-locks:
each such time-lock is specified by a pair $(h,\tau)$, where $h = H(s)$ is a hashed value of a secret ${s}$ and $\tau$ is a time-out value.
In order to unlock this time-lock, $v$ (and only $v$) must provide the value of ${s}$ before time $\tau$. 
If all time-locks of $(u,v)$ are unlocked, $v$ can claim $a$. This automatically triggers arc $(u,v)$.
If any time-lock times out, $a$ is automatically returned to $u$.
We describe a more elaborate hashed time-lock in the next section.

\myparagraph{Properties}
For a swap protocol to be useful, it must guarantee that if all parties follow it then every party ends in an outcome 
at least as favorable as trading all their outgoing for all their incoming assets.
Further, every conforming party should end up with an acceptable outcome, no matter whether other parties follow the protocol or not.
Lastly, rational parties should have no incentive to deviate from the protocol. 
Herlihy~\cite{Herlihy18} captured these properties using the concepts of uniformity and strong Nash equilibrium.
Our definitions, below, are their natural extensions to the more general model of swap systems.


\paragraph{Uniformity}
A swap protocol $\bbP$ is called \emph{uniform} if it satisfies the following two conditions:
\begin{description}
    \item{\emph{Liveness:}} If all parties follow $\bbP$, they all end in outcome \Deal \emph{or better},
			that is $\bbP(v) \succeq \Dealv{v}$ for all $v\in V$.
    \item{\emph{Safety:}} If a party conforms to $\bbP$, then its outcome will be acceptable,
			independently of the behavior of other parties.
\end{description}
A less restrictive concept of uniformity may also be of interest: We say that a protocol $\bbP$
is \emph{weakly uniform} if it satisfies the safety condition above, but the liveness
condition is replaced by the following \emph{weak liveness} requirement: 
if all parties follow $\bbP$, then at least one party ends in an outcome strictly better than \NoDeal.
The assumptions on preference posets imply directly that a protocol that is uniform is also weakly uniform.


\paragraph{Nash equilibria and atomicity}
We extend the concept of outcomes to sets of parties, where an outcome of a set is just a vector of individual outcomes.
On this set we can then define a preference relation in a standard way, via a coordinate-wise ordering of outcomes.
Formally, for any set of parties $C\subseteq V$, an \emph{outcome vector of $C$} is $\vecoutcome = (\outcome_v)_{v\in C}$,
where $\outcome_v\in \Omega_v$ for all $v\in C$. Denote by $\barOmega_C$ the set of all outcome vectors of $C$.
Given two outcome vectors  $\vecoutcome, \vecoutcome'\in \barOmega_C$, we write $\vecoutcome \preceq_C \vecoutcome'$ if
$\outcome_v\preceq_v \outcome'_v$ for all $v\in C$. If also $\vecoutcome\neq \vecoutcome'$ then we write $\vecoutcome \prec_C \vecoutcome'$. 
(In other words, $\vecoutcome \prec_C \vecoutcome'$ means that at least one party in $C$ does strictly better in $\vecoutcome'$ than in $\vecoutcome$,
and every party in $C$ does at least as good.)
In this notation, if all parties follow a protocol $\bbP$,
then the outcome vector $\bbP(C)$ of a protocol $\bbP$ for a set of parties $C$ is $(\bbP(v))_{v\in C}$.

We will say that a protocol $\bbP$ is a \textit{strong Nash equilibrium} if no coalition of participating
parties can improve its vector outcome by deviating from $\bbP$;
more precisely, for every set $C$ of parties, if $\vecoutcome$ denotes the outcome vector of $C$ in
some execution of $\bbP$ where all parties in $V\setminus C$ follow $\bbP$, then we cannot have $\vecoutcome \succ_C \bbP(C)$.
We will call $\bbP$ \emph{atomic} if it is both uniform and a strong Nash equilibrium.


\paragraph{Example~1}
Consider a swap system $\swapSys = (\DG, \prefP)$ whose digraph $\DG$ is shown in Figure~\ref{fig: swap system example}.
The preference poset $\prefP_u$ is generated by two preference pairs $\Dealv{u} \prec \outcomepair{v}{v} \prec \outcomepair{v}{w}$,
the preference poset $\prefP_v$ is generated by two preference pairs $\Dealv{v} \prec \outcomepair{u}{u} \prec \outcomepair{w}{u}$,
and the preference poset $\prefP_w$ is generated by one preference pair $\Dealv{w} \prec \outcomepair{u}{v}$.


\begin{figure}[ht]
\centering
\includegraphics[width=1.75in]{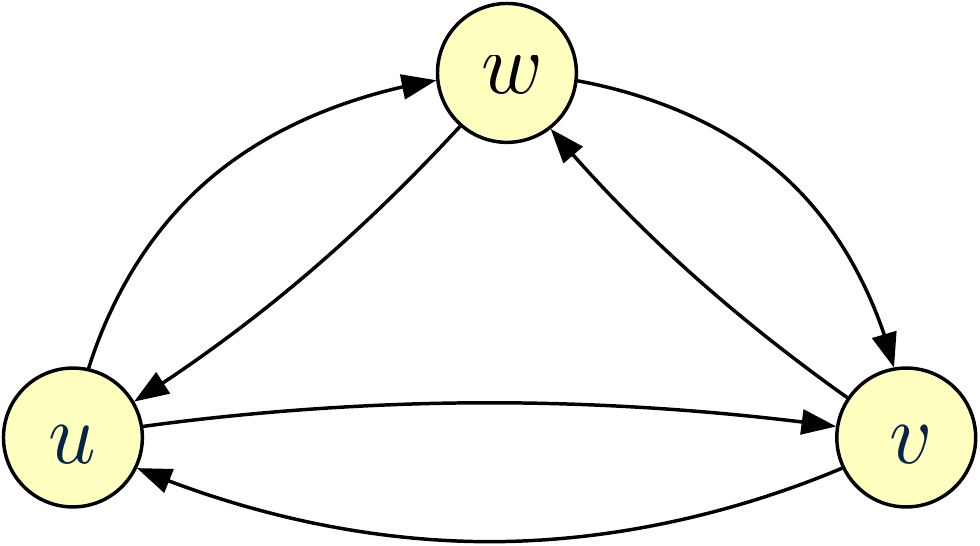} 
\caption{The digraph $\DG$ in the example.}
\label{fig: swap system example}
\end{figure}


Consider also a swap protocol $\bbP$ for $\swapSys$ such that if all parties follow $\bbP$ then all end up
with outcome $\Deal$. Then $\bbP$ is not a strong Nash equilibrium, because for $C = \braced{u,v}$, the parties in $C$
can ignore $\bbP$ altogether and simply swap their assets between themselves, improving their outcomes.
Nevertheless,
as we show later in Section~\ref{sec:uniandnash}, $\swapSys$ does have an atomic protocol.
Roughly, instead of using the whole digraph $\DG$, in this protocol only assets represented by arcs
$(u,w)$, $(w,v)$ and $(v,u)$ will be swapped. Then the outcome of each party will be better than $\Deal$,
and $u$ and $v$ will have no incentive to deviate from this protocol.

%% file: 03_h-swap_systems.tex

In this section, we show that the concept of swap systems is a generalization of Herlihy's model~\cite{Herlihy18}. 
To this end, we define a simple type of swap system called h-swap systems, 
and we show that it captures the model in~\cite{Herlihy18}. 
In particular we prove that in h-swap systems, our definition of atomicity is equivalent to the definition in~\cite{Herlihy18}. 


\myparagraph{h-Swap Systems}
Given a swap system  $\swapSys = (\DG, \prefP)$ and a party $v\in V$, define three sets of outcomes of $v$:
\begin{align*}
    \Discountv{v} \;&=\; \braced{ \outcome \ | \ \outcomein{} = \Ain{v} \wedge \outcomeout{} \neq \Aout{v} }
    \\
    \FreeRidev{v} \;&=\; \braced{ \outcome \ | \ \outcomein{} \neq \varnothing \wedge \outcomeout{} = \varnothing }
	\\
	\Underwaterv{v} \;&=\; \braced{ \outcome \ | \ \outcomein{} \neq \Ain{v} \wedge \outcomeout{} \neq \varnothing }
\end{align*}
Since $\Ain{v}\neq\emptyset$ and $\Aout{v}\neq\emptyset$,
all sets $\Discountv{v}$, $\FreeRidev{v}$ and $\Underwaterv{v}$ are well-defined, 
none of them contains $\NoDealv{v}$ nor $\Dealv{v}$, $\Underwaterv{v} \cap (\Discountv{v}\cup  \FreeRidev{v} ) = \emptyset$,
$\Discountv{v}\cap \FreeRidev{v} = \smbraced{ \outcomepair{\Ain{v}}{\varnothing}}$, and
%
\\
$
    \begin{array}{rcl}
	\Omega_v  &\;=\;& \braced{\NoDealv{v}} \cup \braced{\Dealv{v}} \cup \Discountv{v}
    \\ && \quad\quad \cup  \ \FreeRidev{v} \cup \Underwaterv{v}.
    \end{array}
$ 
\\
%
The inclusive monotonicity property~(p.2) implies that all outcomes in $\FreeRidev{v}$ are better than $\NoDealDv{}{v}$,
and all outcomes in $\Discountv{v}$ are better than $\DealDv{}{v}$. 

We will call $\swapSys$ an \emph{h-swap system} if it satisfies the following conditions for all $v\in V$:
\begin{description}
	\item{(h.1)} If $\omega\in \Underwaterv{v}$ then $\omega\prec_v \NoDealDv{}{v}$,
	\item{(h.2)} Party $v$ has no other non-generic preferences besides these in~(h.1).
\end{description}
In other words, in an h-swap system all preference posets are generated by
relations $\outcome\prec \NoDeal$ for outcomes $\outcome$ in $\Underwater$.
Figure~\ref{fig: h-swap system preferences} illustrates the structure of
a preference poset of an h-swap system\footnote{
This figure differs slightly from Figure~3 in~\cite{Herlihy18}, which mistakenly showed
the sets $\Discountv{v}$ and $\FreeRidev{v}$ as disjoint.}.
Note that in an h-swap system, the set of acceptable outcomes of a node $v$ is
$\acceptsetA_v =  \Omega_v \setminus \Underwaterv{v} 
			= \braced{\NoDealv{v}} \cup \braced{\Dealv{v}}\cup \Discountv{v} \cup  \FreeRidev{v}$.
The preferences of an h-swap system $\swapSys = (\DG, \prefP)$ are uniquely determined by its digraph $\DG$, 
so it is not even necessary to specify $\prefP$.

\begin{figure}
\centering
\includegraphics[width=3.3in]{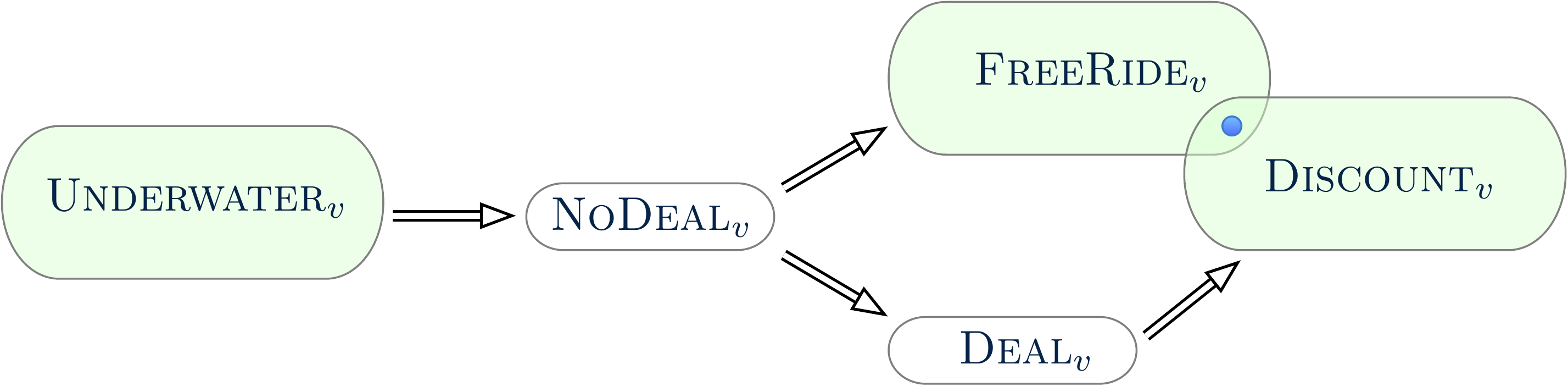} 
\caption{The structure of a preference poset of a party $v$ in an h-swap system. The arrows symbolize the preference relation.
The one outcome in $\Discountv{v}\cap \FreeRidev{v}$ is $\outcomepair{\Ain{v}}{\varnothing}$.
}
\label{fig: h-swap system preferences}
\end{figure}

The preference poset structure of h-swap systems, as defined above, captures the concept of
a party's preferences assumed in the model from~\cite{Herlihy18}, except for the addition of
preferences determined by inclusive monotonicity.


\emph{Comment.}
The model in~\cite{Herlihy18} was not formulated in terms of posets, raising a question of how to
formally capture a relation for pairs of outcomes between which preferences were not specified in~\cite{Herlihy18}.
In our model, such outcomes are considered incomparable in the poset (unless they are related by the inclusive monotonicity).
One may try to consider another option: to allow arbitrary relations between such pairs, providing that the poset
axioms are satisfied and the condition~(h.1) holds. However, with this approach there is no meaningful way to
extend such individual preferences to collective preferences of sets of parties (see the discussion later in this section).


\myparagraph{h-Uniformity} 
To distinguish between our and Herlihy's definition of uniformity, we will refer to his
concept as h-uniformity. A swap protocol $\bbP$ is called \emph{h-uniform} if it satisfies 
the safety property and the following \emph{h-liveness} condition:
If all parties follow $\bbP$, they all end in outcome \Deal. 
This condition seems stricter than our definition of uniformity, but we show that in 
h-swap systems these two definitions are in fact equivalent. In fact, they are also
equivalent to weak uniformity, as defined earlier in Section~\ref{sec:swap systems}.


\begin{lemma}\label{lem: h-uniformity lemma}
Let $\swapSys = (\DG, \prefP)$ be an h-swap system in which some subset of arcs in $\DG$ are triggered, 
and let $Q$ be a path in $\DG$ whose all internal nodes are in acceptable outcomes.
Then, along $Q$, all triggered arcs of $Q$ are before all non-triggered arcs of $Q$.
\end{lemma}

\begin{proof}
If all arcs on $Q$ are triggered, except possibly for the last one, we are done. Otherwise,
let $(x,y)$ be the first non-triggered arc on $Q$ and
let $z$ be the successor of $y$. Since $y$'s outcome is acceptable and $(x,y)$ is not triggered, this outcome must be either
$\NoDealv{y}$ or in $\FreeRidev{y}$. Therefore $(y,z)$ is also not triggered. 
Repeating this argument, we obtain that all arcs on $Q$ after $(x,y)$ are not triggered. 
\end{proof}


\begin{theorem}\label{thm: h-uniformity = uniformity}
Let $\bbP$ be a swap protocol for an h-swap system $\swapSys = (\DG, \prefP)$, where
$\DG$ is strongly connected. Then the following three conditions are equivalent:
(i) $\bbP$ is uniform,
(ii) $\bbP$ is weakly uniform, 
(iii) $\bbP$ is h-uniform.
\end{theorem}

\begin{proof}
Trivially, h-uniformity implies uniformity, which in turn implies weak uniformity. Thus it is sufficient to
show that weak uniformity implies h-uniformity. 

So assume that $\bbP$ is weakly uniform. As the safety condition is the same,
it is sufficient to show that $\bbP$ satisfies the h-liveness property. 
Assume that all parties follow $\bbP$. Then, from the assumptions about safety and weak liveness,
all parties will end up in acceptable outcomes, with at least one party ending in an outcome
strictly better than $\NoDeal$. 

Suppose, towards contradiction, that there is a party with outcome other than $\Deal$. 
This gives us that some arc $(x,y)$ is not triggered. Further, since some party has an outcome other than $\NoDeal$,
there must be a triggered arc $(x',y')$.
By strong connectivity, there is a path $P$ from $x$ to $y'$ whose first arc is $(x,y)$ and the
last arc is $(x',y')$. Then the existence of this path contradicts Lemma~\ref{lem: h-uniformity lemma}.
\end{proof}


\paragraph{h-Atomicity} 
The approach in~\cite{Herlihy18} differs from ours in the way it formalizes the gain of a coalition (subset) of parties when they deviate
from the protocol. Roughly, the definition in~\cite{Herlihy18} captures a collective gain, while our definition views it as a vector
of individual outcomes. In spite of this apparent difference, we
show that in h-swap systems our concept of atomicity is in fact equivalent to the one in~\cite{Herlihy18}. 

In the discussion below, let $\swapSys = (\DG, \prefP)$ be a fixed h-swap system.
Following~\cite{Herlihy18}, we will define the h-outcome of a coalition
$\coalition$ of parties by, in essence, contracting $\coalition$ into a single vertex. 
(The term ``h-outcome'' is ours, to better distinguish this concept from our concept of outcome vectors.)
More formally, define $\coalition$'s incoming and outgoing arcs in a natural way: 
$\Ain{\coalition} = \bigcup_{v \in \coalition} \Ain{v} \setminus \bigcup_{v \in \coalition} \Aout{v}$ 
and similarly, 
$\Aout{\coalition} = \bigcup_{v \in \coalition} \Aout{v} \setminus \bigcup_{v \in \coalition} \Ain{v}$. 
The \emph{h-outcomes} for $\coalition$ are
pairs $\houtcome =  \outcomepair{ \houtcomein{} }{ \houtcomeout{}}$ where
$\houtcomein{} \subseteq \Ain{\coalition}$ and $\houtcomeout{} \subseteq \Aout{\coalition}$. 
$\Omega_\coalition$ is the set of all h-outcomes of $C$.
The preference poset and acceptable set of $\coalition$ are defined analogously to that of a single party in an h-swap system.
That is, we define $\NoDealv{C}$, $\Dealv{C}$, $\Discountv{C}$, $\FreeRidev{C}$, and $\Underwaterv{C}$ in
the natural way,  and we assume the analogues of conditions~(p.1) and~(p.2) for swap systems (in Section~\ref{sec:swap systems})
and conditions~(h.1) and~(h.2) for h-swap systems. The set of acceptable h-outcomes 
$\Ayes{\coalition}$ consists of all h-outcomes of $\coalition$ that are not in $\Underwaterv{C}$.
(Note that if $\coalition$ consists of a single party then its h-outcome is identical to its outcome.)

Define a protocol $\bbP$ to be a \emph{strong Nash h-equilibrium} if it satisfies the following
condition for every set $\coalition$ of parties: providing that the parties outside $\coalition$ follow $\bbP$,
the parties in $\coalition$ cannot end up in an h-outcome better than their outcome resulting from following $\bbP$.
$\bbP$ is called \emph{h-atomic} if it is h-uniform and a strong Nash h-equilibrium.

Culminating the earlier discussion, the following theorem establishes that our model indeed
captures the model introduced in~\cite{Herlihy18}. 


\begin{theorem}\label{thm: atomicity vs h-atomicity}
Let $\bbP$ be a protocol for an h-swap system $\swapSys = (\DG, \prefP)$.
$\bbP$ is atomic if and only if it is h-atomic. 
\end{theorem}

\begin{proof}
$(\Rightarrow)$
Suppose that $\bbP$ is atomic.  Theorem~\ref{thm: h-uniformity = uniformity} implies
that $\bbP$ is h-uniform. Thus, from the definition of h-uniformity,
if all parties follow $\bbP$ then each party's outcome will be $\Deal$.

It remains to show that $\bbP$ is a strong Nash h-equilibrium.
Let $C\subseteq V$, and consider an execution of $\bbP$ in which all parties outside $C$ follow $\bbP$.
Since $\bbP$ is a strong Nash equilibrium, the outcome vector of $C$ is not
better than $(\Dealv{v})_{v\in C}$. Denote by $\houtcome$ the h-outcome
of $C$. We need to show that $\houtcome$ is not better than $\Dealv{C}$.

Towards contradiction, suppose that $\houtcome \succ \Dealv{C}$.
The definition of preference posets for h-outcomes gives us that $\houtcome \in \Discountv{C}$. 
Now consider another execution of $\bbP$ where the parties in $C$ behave just like before, but
they also trigger all arcs connecting two members of $C$.
This will not affect the execution of $\bbP$ for parties outside $C$.
Then the outcome vector $\vecoutcome$ of $C$ consists of all arcs
between $C$ and $V\setminus C$ (in both directions) that are triggered in $\houtcome$,
as well as all arcs with both endpoints inside $C$.
Since $\houtcome \in \Discountv{C}$, each $v\in C$ has all its
incoming arcs in $\vecoutcome$, and there is
at least one $u\in C$ that has one arc to $V\setminus C$ that is not in $\vecoutcome$.
So the outcome of each $v\in C$ is either $\Dealv{v}$ or $\Discountv{v}$, and this $u$'s outcome is $\Discountv{u}$.
But then $\vecoutcome$ is better than $(\Dealv{v})_{v\in C}$, contradicting the assumption
that $\bbP$ is a strong Nash equilibrium.

\smallskip

$(\Leftarrow)$
Now suppose that $\bbP$ is h-atomic; that is, $\bbP$ is h-uniform and is a strong Nash h-equilibrium.
From Theorem~\ref{thm: h-uniformity = uniformity} we obtain that $\bbP$ is uniform.

It remains to prove that $\bbP$ is a strong Nash equilibrium.
Let $C\subseteq V$, and consider some execution of $\bbP$ in which all parties outside $C$ follow $\bbP$.
Since $\bbP$ is a strong Nash h-equilibrium, the h-outcome of $C$ is not better than $\Dealv{C}$. 
We need to show that $C$'s outcome vector is not better than $(\Dealv{v})_{v\in C}$.

We again argue by contradiction. Suppose that $C$'s outcome vector is $\vecoutcome \succ (\Dealv{v})_{v\in C}$.
Then each $v\in C$ has outcome in $\braced{\Dealv{v}}\cup \Discountv{v}$ and there is some
$u\in C$ with outcome in $\Discountv{u}$. This implies that all parties in $C$ have their
incoming arcs in $\vecoutcome$. Further, some outgoing arc of $u$ is not in $\vecoutcome$,
and this arc must go to $V\setminus C$.
We consider the h-outcome of $C$ in the same run of $\bbP$, without changing the behavior of any members of $C$.
(In the h-outcome of $C$ the status of arcs internal to $C$ is not relevant.)
Denote this h-outcome by $\houtcome$.
Then $\houtcome$ will include the same arcs between $C$ and $V\setminus C$ (in both directions) as in $\vecoutcome$.
The properties of $\vecoutcome$ established earlier imply that
$\houtcome \in \Discountv{C}$, and thus $\houtcome \succ \Dealv{C}$, 
which contradicts our earlier assumption that $\bbP$ is a strong Nash h-equilibrium.
\end{proof}


\label{sec: herlihy's protocol}
\input{10_new_appendix_herlihy.tex}

%% file: 10_new_appendix_herlihy.tex
\myparagraph{Herlihy's Protocol}
Herlihy presented a protocol for h-swap systems \cite{Herlihy18} that is h-atomic.
We summarize this protocol that we will refer to as $\hProt$.

Since the generation and distribution of the swap system is not the focus of this paper,
we assume a third-party service that reliably distributes information to the participating parties.
The service begins by assembling a swap graph $\DG$ and distributing it to every party.
Each party $p_i$ then generates and hashes a secret $h_i = \mathsf{hash}(s_i)$ and sends it back to the service.\footnote{
Herlihy describes an optimization where the service computes a feedback vertex set for $\DG$ 
(\emph{i.e.} the removal of this set would leave $\DG$ acyclic).
He refers to these parties as \emph{leaders} and only uses the hashed secrets of these parties in subsequent steps of the protocol. 
As this is not a necessary step, we will ignore it for simplicity.
}
The service distributes the hashed secrets as a vector $h_0...h_n$ to every party.

The protocol can be broken into two phases, which we call contract creation and secret propagation respectively.
The contract creation phase, in essence, realizes $\DG$.
For every arc $(u,v) \in \DG$, party $u$ generates a smart contract with an escrowed asset to counterparty $v$.
Each contract is \emph{hash-locked} by a vector of hashlocks $h_0...h_n$ generated by the given vector of hashed secrets.
A particular hashlock $h_w$ on arc $(u,v)$ unlocks when provided a hashkey $(s_w, p ,\sigma)$,
where $s_w$ is the preimage of $h_w$,
$p$ is any simple path from $v$ to $w$ (where $w$ is the party that generated secret $s_w$),
and $\sigma$ is a sequence of signatures $sig(..,  sig(s_w, w), .., v)$ backwards along path $p$.
It should be noted that a single hashlock may have multiple hashkeys, as any simple path is acceptable.
Hashlocks and hashkeys are also \emph{time-locked}.
Each hashkey only remains valid for a certain amount of time, scaling with the length of the path specified within it.
Ignoring constant factors, a hashkey remains valid for $\barred{p} \cdot \Delta$ time,
where $\Delta$ is an upper bound of the time needed for a single step of a party.
The longer the path in the hashkey, the longer the hashkey remains valid.
A hashlock expires when all of its hashkeys expire, in which the escrowed asset is returned to the sender.
That is, the hashkey containing the longest path from the recipient to the generator of the corresponding secret has expired.
If all hashlocks in the vector are unlocked, the contract triggers and the escrowed asset is sent to the recipient.

When a party $u$ observes that each of its incoming contracts has been created correctly, it enters the secret propagation phase.
Party $u$ first wants to propagate its own secret.
This is done by unlocking hashlock $h_u$ on each of their incoming arcs.
Specifically, $u$ generates hashkey $(s_u, u, sig(s_u, u))$ and uses this to unlock the corresponding hashlock on each arc in $\Ain{u}$.
Party $u$ also wants to propagate the secrets of others, which they learn by observing their own outgoing arcs.
Let $u$ observe on outgoing arc $(u,v)$ that hashlock $h_w$ was unlocked by hashkey $(s_w, p, \sigma)$.
Then $u$ can generate hashkey $(s_w, u+p, sig(\sigma, u))$ and unlock the corresponding hashlock $h_w$ on each arc in $\Ain{u}$.

\begin{algorithm}
    \caption{Herlihy's Protocol For Vertex $v$}\label{alg:h proto}
    \textbf{Input:} Digraph $\DG$, vector $\angled{h_0,...,h_n}$, secret $s_v$
    \begin{algorithmic}[1]
    \For{every $(v,w) \in \Aout{v}$}         \Comment{Phase 1}
        \State create contract to $w$ hashlocked by $\angled{h_0,...,h_n}$
    \EndFor
    \State \textbf{upon} contract for every $(u,v) \in \Ain{v}$ \Comment{Phase 2}
    \State generate hashkey $k_1 = (s_v, v, sig(s_v, v))$
    \For{every arc $(u,v) \in \Ain{v}$}         
        \State unlock hashlock $h_v$ with $k_1$
    \EndFor
    \While{no timed out hashlock \textbf{and} not all assets received}
        \If{new hashkey $k_2 = (s,p,\sigma)$ on $(v,w) \in \Aout{v}$}
            \State generate hashkey $k_3 = (s, v+p, sig(\sigma, v))$
            \For{every arc $(v,w) \in \Ain{v}$}         
                \State unlock hashlock with $k_3$
            \EndFor
\EndIf
\EndWhile
\end{algorithmic}
\end{algorithm}

\paragraph{Example~2}
Consider the swap graph in Figure~\ref{fig:herlihy not strong nash}. 
Assume each party is given the same vector of hashed secrets.
Parties start the protocol by creating their outgoing contracts using this vector.
The party $x$ creates the contract $(x, v)$,
the party $u$ creates the contracts $(u, x)$ and $(u, v)$,
the party $v$ creates the contracts $(v, u)$ and $(v, y)$,
and
the party $y$ creates the contract $(y, u)$.
Then when the party $x$ observes that the contract $(u,v)$ is created,
it releases its secret on $(u,x)$.
Once the party $u$ observes this secret on its outgoing contract $(u,x)$,
it applies it to its incoming contracts $(v,u)$ and $(y,u)$.
Similarly, when the parties $v$ and $y$ observe the secret on their outgoing contracts,
they apply it to their incoming contracts $(x,v)$ and $(v,y)$ respectively.
Thus, the secret $s_x$ is propagated through the whole graph.
With similar steps,
each other party releases its secret on its incoming contracts,
and each secret is propagated by other parties to the rest of the graph.
Thus, all secrets are eventually applied to all contracts, and all assets are transferred.


\begin{lemma}
\label{obs: herlihy observation}
Consider the execution of $\hProt$ on a directed graph $\DG$. Let $u$ be a party that follows
$\hProt$. After the execution is complete, for any $e\in\Aout{u}$ of $u$, $e$ is triggered only if all arcs in $\Ain{u}$ are triggered.
\end{lemma}

If $e$ is triggered, every hashlock on $e$ was unlocked.
Since $u$ is following $\hProt$, it will observe whenever a hashlock on $e$ is unlocked.
Whenever a hashlock $h_i$ is unlocked, $u$ sees a hashkey $k_1 = (s_i, p, \sigma)$.
Then, $u$ can generate a hashkey $k_2 = (s_i, u+p, sig(\sigma, u))$ to post on the corresponding hashlocks $h_i$ for their incoming arcs.
Since $k_1$ was an acceptable hashkey, and increasing the length of a hashkey by 1 means it remains acceptable for $\Delta$ more time, $u$ had sufficient time to post $k_2$ to all arcs in $\Ain{u}$.
We repeat this argument for every hashlock on $e$, 
wherein when every hashlock on $e$ is unlocked, every hashlock on every arc in $\Ain{u}$ is also unlocked.

%% file: 04_characterization.tex


As shown in~\cite{Herlihy18}, all swap systems considered in Herlihy's approach (that is all h-swap systems, in our terminology)
have an atomic protocol, providing that the underlying digraph is strongly connected.
In our more general model this is not always the case. Consider, for example, 
a swap system whose digraph is shown in Figure~\ref{fig:herlihy not strong nash}.
The only non-generic preferences are: $\Dealv{v} \prec \outcomepair{u}{u}$ for $v$, and $\Dealv{u} \prec \outcomepair{v}{v}$ for $u$.
Using the liveness condition for $x$ and $y$, any atomic protocol needs to trigger arcs $(u,x)$, $(x,v)$,
$(v,y)$ and $(y,u)$. But
$u$ and $v$ can cooperatively deviate from the protocol by triggering only arcs $(u,v)$ and $(v,u)$,
each obtaining a better outcome than if they followed the protocol. So this protocol cannot be a strong
Nash equilibrium, and thus is not atomic.


\begin{figure}[ht]
\centering
\includegraphics[width=1.6in]{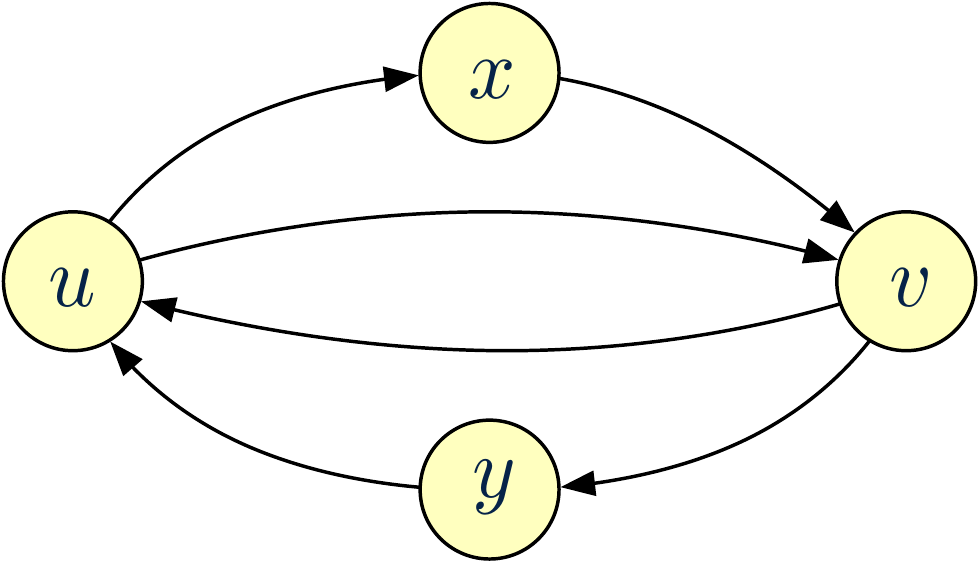} 
\caption{The example of a swap system in which $\hProt$ is not a strong Nash equilibrium.}
\label{fig:herlihy not strong nash}
\end{figure}


This raises the question as to whether there exists a simple characterization of swap systems that
admit atomic protocols. We provide such a characterization in this section.
Interestingly enough, we show that if a swap system admits an atomic protocol, 
then it also admits an atomic protocol that is essentially equivalent to running
Herlihy's protocol  on a suitable subgraph.

We saw Herlihy's protocol~\cite{Herlihy18}, denote by $\hProt$, in the previous section.
Herlihy proved that $\hProt$ is h-atomic for h-swap systems (in our terminology). 


\myparagraph{Uniformity and Nash equilibrium of Herlihy's protocol} 
Let $\swapSys= (\DG, \prefP)$ be any swap system with strongly connected digraph $\DG$.
If all parties follow $\hProt$ then they all will end up in outcome $\Deal$.
If $v$ follows $\hProt$ then either it does not trigger any outgoing arcs, and thus its outcome is in
$\smbraced{\NoDealv{v}} \cup \FreeRidev{v}$, or it triggers some, but then also all its incoming arcs are triggered,
so its outcome is in $\smbraced{\Dealv{v}} \cup \Discountv{v}$. In each case, regardless of the behavior of
other parties, this outcome is at least as good
as $\NoDealv{v}$, and thus acceptable. This means that $\hProt$ is uniform.

We now claim that $\hProt$ is a \emph{Nash equilibrium} in $\swapSys$, in the sense that no single party can improve
its outcome by deviating from $\hProt$, if all other parties follow $\hProt$. 
If all parties follow the protocol, all outcomes are $\Deal$. If any party $v$ has outcome
$\outcome \succ \Dealv{v}$, then in  $\outcome$ there needs to be an incoming arc $(u,v)$ but some outgoing arc
$(v,w)$ must be missing. (This holds for any preference poset, by properties~(p.1) and~(p.2).)
In Herlihy's protocol a vertex triggers an outgoing arc only if all its incoming arcs are triggered. 
So if all parties other than $v$ follow the protocol, then we get a contradiction by considering a path from $w$ to $u$
(following an argument similar to the proof of Theorem~\ref{thm: h-uniformity = uniformity}). 

Note that this argument does not work for larger coalitions. In fact, in the swap system example
discussed earlier in this section, $\hProt$, nor any other protocol, is a strong Nash equilibrium.

\smallskip

\myparagraph{Characterization}
Let $\swapSys = (\DG, \prefP)$ be a swap system for a set of parties $V$,  and let $\GG$ and $\HG$ be two subgraphs of $\DG$. 
$\GG$ will be called \emph{piece-wise strongly connected} if every connected component of $\GG$ is strongly connected.
$\GG$ is called \emph{spanning} if its vertex set is $V$. If $\GG$ is spanning, we say that $\HG$ \emph{dominates} $\GG$ if
$\DealDv{\HG}{v} \succeq \DealDv{\GG}{v}$ for all vertices $v$ in $\HG$.
In other words, if only the arcs in $\HG$ are triggered, then all parties in $\HG$ end in outcomes at least as good as if all their arcs in $\GG$ were triggered.
Also, $\HG$ \emph{strictly dominates} $\GG$ if, in addition, there exists a party $u$ of $\HG$ such that 
$\DealDv{\HG}{u} \succ \DealDv{\GG}{u}$.  
That is, every party in $\HG$ ends in an outcome at least as good and 
at least one party strictly improves their outcome when triggering the arcs of $\HG$ instead of $\GG$.

For example, consider the swap system $\swapSys = (\DG, \prefP)$ in Example~1. 
The subgraph $\GG_1 = \DG$ is spanning, and is strictly dominated by the subgraph $\HG$ consisting of vertices $u,v$ and arcs $(u,v)$ and $(v,u)$.
On the other hand, the subgraph $\GG_2$ that has arcs $(u,w)$, $(w,v)$ and $(v,u)$ is
spanning, and there is no subgraph of $\DG$ that strictly dominates it.


\begin{theorem}
\label{thm:uni-and-snash}
A swap system $\swapSys = (\DG, \prefP)$ has an atomic swap protocol if and only if 
there exists a spanning subgraph $\GG$ of $\DG$ with the following properties:
{(c.1)} $\GG$ is piece-wise strongly connected and has no isolated vertices,
{(c.2)} $\GG$ dominates $\DG$, and
{(c.3)} no subgraph $\HG$ of $\DG$ strictly dominates $\GG$.
\end{theorem}


\begin{proof}
$(\Rightarrow)$
Let $\bbP$ be an atomic swap protocol for $\swapSys$. Define
$\GG$ to be the subgraph whose vertex set is $V$ and whose arcs are the arcs triggered in an execution of $\bbP$
where all parties follow the protocol. 
By definition of $\bbP$'s atomicity, $\GG$ is spanning.


We first show property (c.1). First, $\GG$ cannot have any isolated
vertices, since any isolated vertex $v$ of $\GG$ would have outcome $\NoDealDv{\DG}{v}$ when all parties follow $\bbP$.
This would contradict the uniformity (the liveness condition) of $\bbP$.
Second, if $\GG$ had a connected component $B$ that is not strongly connected, then 
$B$ would contain a strongly connected component $C$ of $\GG$ that has no arcs of $\GG$
coming from $V\setminus C$ but has at least one arc of $\GG$ going to $V\setminus C$.
We could then consider another run of $\bbP$ in which 
the parties in $C$ ignore $\bbP$ entirely and simply trigger the arcs of $\GG$ that are within $C$.
By the inclusive monotonicity property~(p.2) of swap systems, this would strictly improve
the outcome vector of $C$, contradicting $\bbP$ being a strong Nash equilibrium. 
We can thus conclude that such $B$ cannot exist, completing the proof that $\GG$ is piece-wise strongly connected.


Next, we consider property~(c.2). By the uniformity  (liveness)
of $\bbP$, every party $v$ must end in outcome $\DealDv{\DG}{v}$ or better when all parties follow $\bbP$.
The arcs that are triggered at the conclusion of $\bbP$ are exactly the arcs in $\GG$.
Therefore $\DealDv{\GG}{v} \succeq \DealDv{\DG}{v}$, for all parties $v$.


Finally, we show property~(c.3). Suppose there is a subgraph $\HG$ that strictly dominates $\GG$, towards contradiction.
Let $C$ be the set of vertices of $\HG$. Modify the behavior of the parties in $C$ to ignore $\bbP$ and
instead trigger exactly the arcs of $\HG$, giving $C$ the outcome vector $(\DealDv{\HG}{v})_{v\in C}$.
Then, $\bbP(C) = (\DealDv{\GG}{v})_{v\in C} \prec_{C} (\DealDv{\HG}{v})_{v\in C}$,
as $\HG$ strictly dominates $\GG$.
This contradicts the assumption that $\bbP$ is a strong Nash equilibrium, proving that $\HG$ does not exists.


$(\Leftarrow)$
Suppose that $\GG$ is a spanning subgraph that satisfies properties (c.1), (c.2) and~(c.3). We show that then there is an
atomic protocol for $\swapSys$. 

Let $\swapSys^\GG$ be the h-swap system with digraph $\GG$.
Our protocol, denoted $\hProt_\GG$, simply executes Herlihy's protocol $\hProt$ on $\swapSys^\GG$.
For simplicity, assume that $\GG$ is strongly connected; otherwise we can apply our reasoning below
to each strongly connected component of $\GG$ separately.
By the h-liveness condition of $\hProt_\GG$, if all parties follow $\hProt_\GG$ then each will end up in outcome $\DealDv{\GG}{}$.
Also, any party $v$ that follows $\hProt_\GG$ will not have any of its arcs outside $\GG$
triggered and, by the safety property of $\hProt_\GG$, will end up in an outcome that is acceptable in $\swapSys^\GG$, 
that is in
$\smbraced{\NoDealDv{\GG}{v}} \cup \FreeRideDv{\GG}{v} \cup \smbraced{\DealDv{\GG}{v}} \cup \DiscountDv{\GG}{v}$.

When comparing outcomes in the argument that follows, we will use notation ``$\prec$'' for the preference 
relation in the original swap system $\swapSys$ (that is, \emph{not} in the auxiliary system $\swapSys^\GG$). 
Similarly, unless stated otherwise, the term ``acceptable'' also refers to the acceptability of an outcome in $\swapSys$.


We first show that $\hProt_\GG$ is uniform. Suppose that every party follows $\hProt_\GG$. 
Then, by the h-uniformity of $\hProt_\GG$, the outcome of each party $v$ will be $\DealDv{\GG}{v}$. 
Using the assumptions that $\GG$ is spanning and that it dominates $\DG$, we obtain that
$\DealDv{\GG}{v} \succeq \DealDv{\DG}{v}$ for all parties $v$, so $\hProt_\GG$ indeed satisfies the liveness condition.


Next, we deal with the safety condition.
Using the properties of $\hProt_\GG$ established above, if a party $v$ conforms to $\hProt_\GG$ then
we have two cases. Either the outcome $\outcome$ of $v$ satisfies
$\outcome \in \smbraced{\NoDealDv{\GG}{v}} \cup \FreeRideDv{\GG}{v}$, in which case
$\outcome \in \smbraced{\NoDealDv{\DG}{v}} \cup \FreeRideDv{\DG}{v}$ as well (because no edges of $v$ outside $\GG$ are triggered),
so $\outcome\succeq \NoDealDv{\DG}{v}$, that is $\outcome$ is acceptable. 
Or $\outcome \in \smbraced{\DealDv{\GG}{v}} \cup \DiscountDv{\GG}{v}$, in which case,
using the monotonicity property~(p.2) for $\swapSys$ and assumption~(c.2), we obtain
$\outcome \succeq \DealDv{\GG}{v}\succeq \DealDv{\DG}{v} \succ \NoDealDv{\DG}{v}$; that is $\omega$ is acceptable in this case as well.
We conclude that $\hProt_\GG$ satisfies the safety property, completing the proof that $\hProt_\GG$ is uniform.


It remains to show that $\hProt_\GG$ is a strong Nash equilibrium for $\swapSys$. Assume that it is not, towards contradiction.
Then there exists a coalition $C \subseteq V$ that, by deviating from $\hProt_\GG$,
can end in an outcome vector $\vecoutcome \succ (\DealDv{\GG}{v})_{v \in C}$, even though all parties outside $C$ follow $\hProt_\GG$.
We can assume $C$ is maximal, in the sense that each party outside of $C$ ends in an outcome that is not 
$\DealDv{\GG}{}$  nor in $\DiscountDv{\GG}{}$. 
Otherwise, we can add those parties to $C$ and the relation $\vecoutcome \succ (\DealDv{\GG}{v})_{v \in C}$ will be preserved.

We first show that no arc $(u,v)\in A$ entering $C$ from outside (that is $u \in V \setminus C$ and $v \in C$)  is triggered.
Assume such an arc is triggered, towards contradiction.
Firstly, $(u,v)$ must be in $\GG$, otherwise $u$ would not be following $\hProt_\GG$ by creating/triggering this arc.
By the h-safety property of $\hProt$ in $\swapSys^\GG$, 
$\hProt_\GG$ guarantees that $u$ must end up in an outcome acceptable in $\swapSys^\GG$.
This means that
$u$'s outcome is in $\smbraced{\DealDv{\GG}{u}}\cup \DiscountDv{\GG}{u}$, contradicting the assumption that $C$ is maximal.
So, indeed, $(u,v)$ cannot be triggered.

Further, without loss of generality we can assume that no arc from $C$ to $V\setminus C$
is triggered. This is because, as we just showed, for each $v \in C$, $v$ only receives arcs from other parties in $C$.
Then, no member of $C$ can have its outcome worsened if $v$ changes its behavior and does not trigger any arc to $V\setminus C$.

Thus all arcs that appear in $\vecoutcome$ are between members of $C$. 
Let $\HG$ be the subgraph with vertex set $C$ and the arcs that are in $\vecoutcome$, that is $\vecoutcome =  (\DealDv{\HG}{v})_{v \in C}$.
Since $\vecoutcome \succ (\DealDv{\GG}{v})_{v \in C}$, then $\DealDv{\HG}{v} \succeq \DealDv{\GG}{v}$ for all $v\in C$
and $\DealDv{\HG}{w} \succ \DealDv{\GG}{w}$ for some $w\in C$.
This means that $\HG$ strictly dominates $\GG$, contradicting~(c.3).
We conclude that no such $C$ exists, and thus $\hProt_\GG$ is a strong Nash equilibrium protocol.
\end{proof}

\begin{figure*}
\centering
\includegraphics[width=6in]{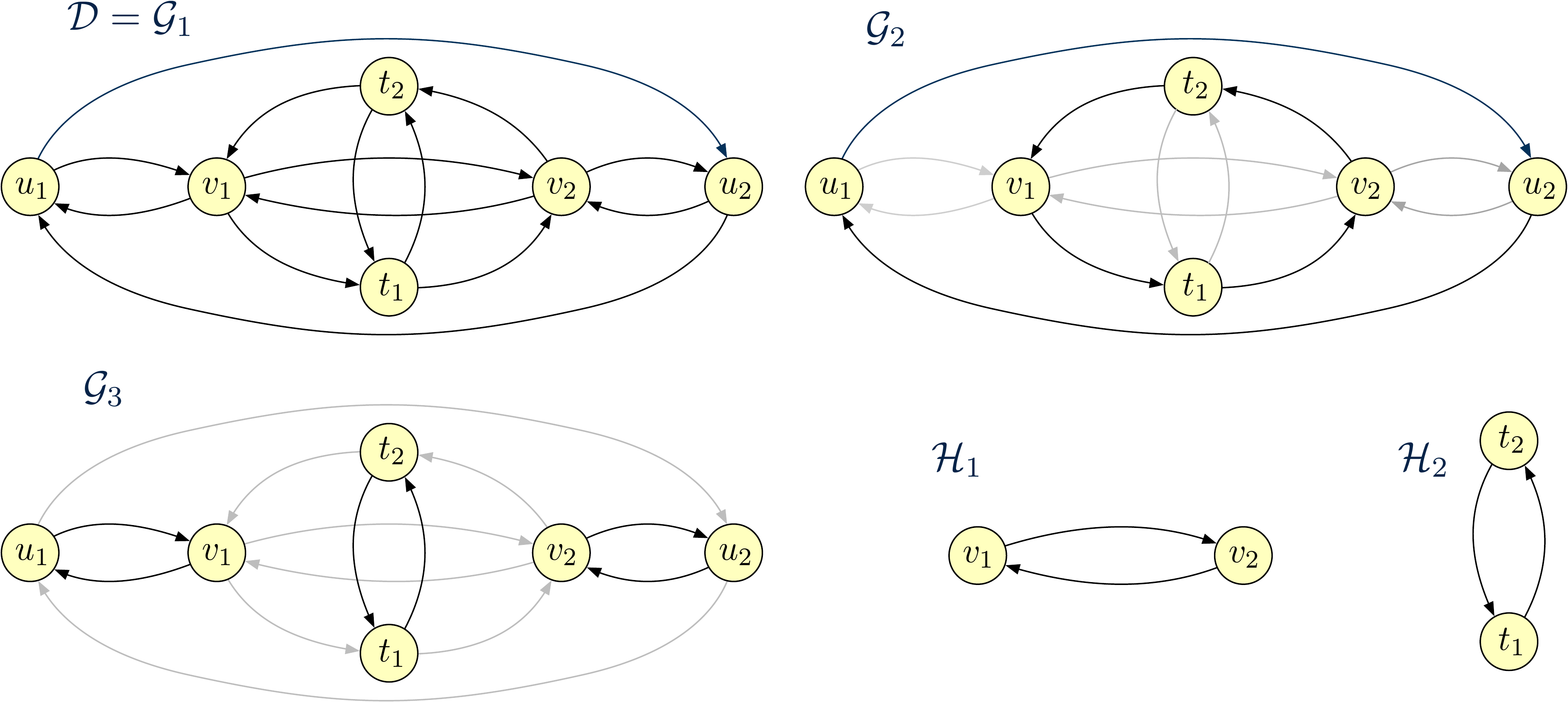} 
\caption{Digraph $\DG$ in the example illustrating Theorem~\ref{thm:uni-and-snash}.}
\label{fig:multi candidates}
\end{figure*}

\smallskip
\emph{Comment:} As some readers may have noticed, the proof of the $(\Rightarrow)$ implication in Theorem~\ref{thm:uni-and-snash}
does not use the safety property of protocol $\bbP$. What this shows, in essence, is that in our setting of swap systems,
a swap protocol that
has the liveness and strong Nash equilibrium properties can be modified to also satisfy the safety property. 
\smallskip

With Theorem~\ref{thm:uni-and-snash} established,
we can determine if a given swap system $\swapSys$ permits an atomic protocol.
Additionally, if it does, we can define such a protocol.
Algorithm~\ref{alg:gen-proto} describes how to check if a given swap protocol permits an atomic protocol.
If it does not, it returns -1.
Otherwise, it returns a set of strongly connected components.
Then, the atomic swap protocol is running $\hProt$ with each component as the underlying graph.

\begin{algorithm}
    \caption{Generalized Atomic Swap Protocol}\label{alg:gen-proto}
    \textbf{Input:} Swap System $\swapSys = (\DG, \prefP)$ \\
    \textbf{Output:} set of strongly connected components \textbf{or} -1
    \begin{algorithmic}[1]
    \For{every spanning subgraph $\GG$ of $\DG$}
        \If{$\GG$ is piece-wise strongly connected}
        \Comment{c.1}
            \If{$\GG$ dominates $\DG$}
            \Comment{c.2}
                \For{every subgraph $\HG$ of $\DG$}
                \Comment{c.3}
                    \If{$\HG$ strictly dominates $\GG$}
                        \State \Return -1
                    \EndIf
                \EndFor
                \State \Return $\braced{C \;|\; C \text{ is an SCC of } \GG}$
            \EndIf
        \EndIf
    \EndFor
    \State \Return -1
\end{algorithmic}
\end{algorithm}


\medskip

\paragraph{Example~3}
\emph{(Example~1 continued)}
To illustrate Theorem~\ref{thm:uni-and-snash}, consider again the swap system $\swapSys = (\DG, \prefP)$ in Example~1. 
Let $\GG_1 = \DG$. Then $\GG_1$ is spanning, satisfies conditions~(c.1) and~(c.2), but it does not satisfy
condition~(c.3) because it is strictly dominated by subgraph $\HG$ consisting of vertices $u,v$ and
arcs $(u,v)$ and $(v,u)$. .

But we can take instead subgraph $\GG_2$ consisting of arcs $(u,w)$, $(w,v)$ and $(v,u)$.
Then $\GG_2$ is spanning, satisfies~(c.1) and~(c.2), and there is no subgraph of $\DG$ that strictly
dominates $\GG_2$, so~(c.3) holds as well. Therefore $\swapSys$ has an atomic swap protocol.
This protocol will completely ignore arcs outside $\GG_2$. It will execute Herlihy's protocol,
described in Section~\ref{sec:h-swap systems}, using $\GG_2$ as the underlying graph. In fact, since $\GG_2$ is a simple cycle, 
the full protocol of Herlihy is not needed --- the simpler protocol that uses only one leader
and does not need signatures can be used instead, see~\cite{Herlihy18}.  Roughly, this
protocol chooses one leader, say $u$, who creates its secret, then the smart contracts
based on this (hashed) secret are created by parties $u$, $v$ and $w$ on their outgoing arcs,
with decreasing time-out values, and then the counter-parties claim the assets in these
contracts, one by one, in the reverse order.  


\medskip

\paragraph{Example~4}
We now give a larger example. Consider the swap system $\swapSys = (\DG, \prefP)$ with
digraph $\DG$ in~\autoref{fig:multi candidates} and with preference posets defined as follows.
For $i=1,2$:
\begin{itemize}[leftmargin=0.1in]
    \item The preference poset of ${u_i}$ is generated by
	 $\Dealv{u_i} \prec \outcomepair{v_{i}}{v_{i}}$ and  $\Dealv{u_i} \prec \outcomepair{u_{3-i}}{u_{3-i}}$.
	 \item The preference poset of ${v_i}$ is generated by $\Dealv{v_i} \prec \outcomepair{u_{i}}{u_{i}}$ and 
	 	 $\Dealv{v_i} \prec \outcomepair{t_{3-i}}{t_{i}} \prec \outcomepair{v_{3-i}}{v_{3-i}}$.
	\item The preference poset of $t_i$ is generated by
	 $\Dealv{t_i} \prec \outcomepair{v_{i}}{v_{3-i}}$ and $\Dealv{t_i} \prec \outcomepair{t_{3-i}}{t_{3-i}}$.
\end{itemize}

We consider three candidates for the spanning sugraph $\GG$.
One candidate is $\GG_1 = \DG$. It's obviously
spanning and it satisfies conditions~(c.1) and~(c.2) from Theorem~\ref{thm:uni-and-snash}.  
However, it is strictly dominated by several subgraphs including the following two:
subgraph $\HG_1$ consisting of $v_1$ and $v_2$ with arcs $(v_1,v_2)$ and $(v_2,v_1)$,
and subgraph $\HG_2$ consisting of $t_1$ and $t_2$ with arcs $(t_1,t_2)$ and $(t_2,t_1)$.

Another candidate $\GG_2$ consists of arcs $(u_i,u_{3-i})$, $(v_i,t_i)$ and $(t_{i},v_{3-i})$. for $i=1,2$.
$\GG_2$ is spanning and it satisfies condition~(c.1). By inspecting the preferences of each vertex,
it also satisfies~(c.2). But it is strictly dominated by $\HG_1$.

The third candidate $\GG_3$ consists of arcs $(u_i,v_i)$, $(v_i,u_i)$, and $(t_i,t_{3-i})$, for $i = 1,2$.
It is also spanning and satisfies conditions~(c.1) and~(c.2). 
Also, the outcome of each vertex in $\GG_3$ is maximal in its preference poset, so
there is no subgraph of $\DG$ that strictly dominates $\GG_3$. Thus, by Theorem~\ref{thm:uni-and-snash}, 
$\swapSys$ has an atomic swap protocol. This protocol is obtained by running $\hProt$ in $\GG_3$.


%% file: 05_np_hardness.tex

Now that we have characterized the swap systems that permit an atomic protocol, 
a natural next question is the complexity of the corresponding decision problem:
given a swap system, does it permit an atomic protocol?
In the next two sections, we consider this.
In this section, we define the corresponding decision problem and show it is in $\NP$-hard.
In the following section, we tighten this classification to $\SigmaTwoP$-complete.
Although showing the problem is $\SigmaTwoP$-complete would imply it is $\NP$-hard,
we first present the $\NP$-hardness proof since it is more digestible,
and then present the more involved $\SigmaTwoP$-completeness proof.

Let $\swapAtomic$ be the following decision problem: The input is a swap system $\swapSys = (\DG, \prefP)$,
where $\DG$ is a (weakly) connected digraph with no vertices of in-degree or out-degree $0$. The objective
is to decide whether $\swapSys$ has an atomic swap protocol.

\begin{theorem}
\label{thm:nphard}
$\swapAtomic$ is $\NP$-hard, even for swap systems $\swapSys = (\DG, \prefP)$ in which digraph $\DG$ is strongly connected.
\end{theorem}

\begin{proof}
The proof is by showing a polynomial-time reduction from $\CNF$. Recall that in $\CNF$
we are given a boolean expression $\alpha$ in conjunctive normal form, and
the objective is to determine whether there is a truth assignment that satisfies $\alpha$. 
In our reduction we convert $\alpha$ into a swap system $\swapSys = (\DG, \prefP)$
such that $\alpha$ is satisfiable if and only if  $\swapSys$ has an atomic swap protocol.

Let $x_1,x_2,...,x_n$ be the variables in $\alpha$.  The negation of $x_i$ is denoted $\barx_i$.
We will use notation $\tildex_i$ for an unspecified literal of variable $x_i$,
that is $\tildex_i \in \braced{x_i,\barx_i}$. Let $\alpha = c_1 \vee c_2 \vee ... \vee c_m$, where each $c_j$ is a clause.
Without loss of generality we assume that in each clause all literals  are different.

We first describe the reduction. 
The fundamental approach is similar to other reductions from $\CNF$: $\DG$ will consist of
so-called ``gadgets'' which are subgraphs used to simulate the role of variables and clauses.
We then add some connections between these gadgets and specify appropriate preference pairs to
assure that the resulting swap system $\swapSys$ satisfies the required property, given above.

Specifically, in $\DG$ there will be $n$ gadgets corresponding to variables, 
 $m$ gadgets corresponding to clauses, and one more vertex called the \emph{core vertex}.
The {\emph{$x_i$-gadget}} has vertices $x_i$ and $\barx_i$, that represent the literals of variable $x_i$,
plus two additional vertices $s_i$ and $t_i$. 
Its internal arcs are $(s_i,x_i)$, $(s_i,\barx_i)$, $(s_i,t_i)$, and $(t_i,s_i)$.
The {\emph{$c_j$-gadget}} has two vertices $c_j$ and $a_j$, with internal arcs $(c_j,a_j)$, $(a_j,c_j)$.
(See Figure~\ref{fig: np-hardness gadgets}.) We then add the following arcs:
For each clause $c_j$ and each literal $\tildex_i$ in $c_j$, we add arc $(\tildex_i,c_j)$.
The core vertex, denoted $b$, is connected by arcs to and from each vertex in the above gadgets.


\begin{figure*}
\centering
\includegraphics[width=4.2in]{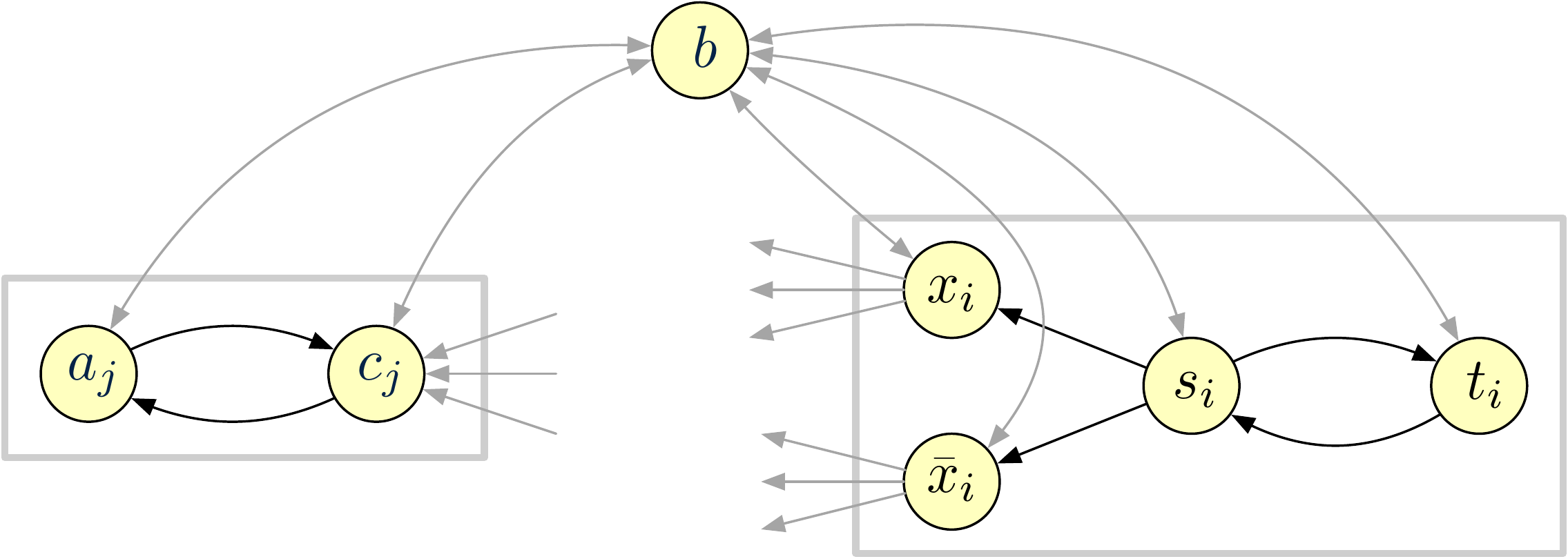} 
\caption{The variable and clause gadgets in the proof of Theorem~\ref{thm:nphard}. The arcs to and from 
the core vertex $b$ are shown as bi-directional arcs.}
\label{fig: np-hardness gadgets}
\end{figure*}


Next, we describe the preference posets $\prefP_v$, for each vertex $v$ in $\DG$. 
(See Figure~\ref{fig: np-hardness preferences}.) As explained in Section~\ref{sec:swap systems},
an outcome $\outcomepair{\outcomein{}}{\outcomeout{}}$
of a vertex $v$ is specified by lists $\outcomein{}$ and $\outcomeout{}$ of its in-neighbors and out-neighbors,
respectively, and any preference poset can be uniquely defined by an appropriate set of generators.

The center vertex $b$'s preference poset is generated by all relations $\outcome \prec \NoDealv{b}$,
for $\outcome\in\Underwaterv{b}$. (This is the same poset as in h-swap systems.)
For each vertex $s_i$, its preference poset is generated by relations
	$\outcomepair{b,t_i}{t_i,\barx_i} \prec \outcomepair{t_i}{t_i}$,
	$\outcomepair{t_i}{t_i}  \prec \outcomepair{b}{b,x_i}$,
	and
	$\outcomepair{t_i}{t_i} \prec \outcomepair{b}{b,\barx_i}$. 
For each vertex $t_i$, its generators are $\Dealv{t_i} \prec \outcomepair{b}{b}$ and
$\outcomepair{b,s_i}{b} \prec \outcomepair{s_i}{s_i}$.  Each vertex $\tildex_i\in\braced{x_i,\barx_i}$
has one generator $\Dealv{\tildex_i}\prec \outcomepair{b}{b}$.
Each vertex $c_j$ has generators
$\Dealv{c_j} \prec \outcomepair{b,\tildex_i}{b}$, for each literal $\tildex_i$ in $c_j$.
The only generator of each vertex $a_j$ is $\Dealv{a_j} \prec \outcomepair{b}{b}$.

With this, the description of $\swapSys$ is complete. The construction of $\swapSys$ clearly takes time
that is polynomial in the size of $\alpha$.


\begin{figure*}
\centering
\includegraphics[width=5.75in]{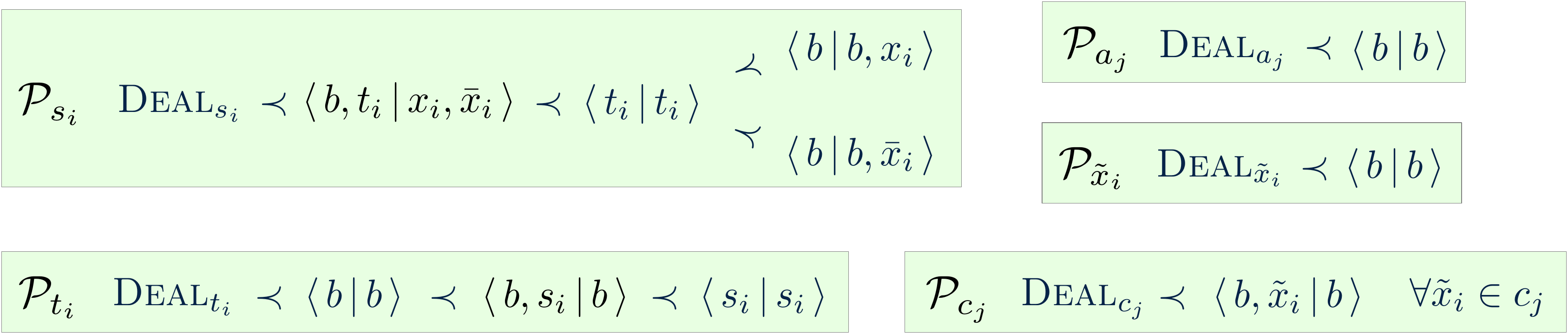} 
\caption{The specifications of preference posets in the reduction. 
We include in the figure some generic preferences for $s_i$ and $t_i$, to
illustrate the relationship of some outcomes to the corresponding $\Deal$ outcome.
}
\label{fig: np-hardness preferences}
\end{figure*}


\smallskip

Applying Theorem~\ref{thm:uni-and-snash},  it remains to show that
$\alpha$ is satisfiable if and only if $\DG$ has a spanning subgraph $\GG$ with the following properties:
{(c.1)} $\GG$ is piece-wise strongly connected and has no isolated vertices,
{(c.2)} $\GG$ dominates $\DG$, and
{(c.3)} no subgraph $\HG$ of $\DG$ strictly dominates $\GG$.

\smallskip

$(\Rightarrow)$
Suppose that $\alpha$ is satisfiable, and consider some satisfying assignment for $\alpha$.
Using this assignment, we construct a spanning subgraph $\GG$ of $\DG$ that satisfies the three conditions (c.1)-(c.3) above.

Graph $\GG$ will contain all vertices from the above construction and all arcs that connect
$b$ to all other vertices, in both directions. This makes $\GG$ spanning and strongly connected,
so (c.1) holds. Other arcs of $\GG$ are defined as follows.
For each true literal $\tildex_i$, add to $\GG$ arc $(s_i,\tildex_i)$.
For each clause $c_j$ and each true literal $\tildex_i$ in $c_j$, add to $\GG$ arc $(\tildex_i,c_j)$. 

\smallskip

Condition~(c.2) can be verified through routine inspection,
by observing that $\DealDv{\DG}{v}\preceq \DealDv{\GG}{v}$ holds for each vertex $v$, directly
from the above specification of the preference posets. 
For example, for a vertex $s_i$, if $\tildex_i$ is the true literal of $x_i$ then
we have $\DealDv{\GG}{s_i} = \outcomepair{b}{b,\tildex_i} \succeq \DealDv{\DG}{s_i}$.
For any $\tildex_i$, let $C(\tildex_i)$ be the set of clauses that contain $\tildex_i$. If $\tildex_i$ is true then 
$\DealDv{\GG}{\tildex_i} = \outcomepair{b,s_i}{b,C(\tildex_i)} = \DealDv{\DG}{\tildex_i}$,
and if $\tildex_i$ is false then $\DealDv{\GG}{\tildex_i} = \outcomepair{b}{b} \succeq \DealDv{\DG}{\tildex_i}$.
For each clause $c_j$, denote by $T(c_j)$ the set of true literals in $c_j$. Since we use a satisfying
assignment, each $T(c_j)$ is non-empty. For a clause $c_j$, for any true literal $\tildex_i \in T(c_j)$, applying
the inclusive monotonicity property~(p.2) we have
$\DealDv{\GG}{c_j} = \outcomepair{b,T(c_j)}{b}  \succeq \outcomepair{b,\tildex_i}{b} \succeq \DealDv{\DG}{c_j}$.

\smallskip

It remains to verify condition~(c.3). Let $\HG$ be a subgraph of $\DG$, and suppose that
$\HG$ dominates $\GG$, that is $\DealDv{\HG}{v} \succeq \DealDv{\GG}{v}$ for all vertices $v$ in $\HG$.

We claim first that $\HG$ must contain $b$. Indeed, otherwise for any vertex $v$ of $\HG$
we would have $\DealDv{\HG}{v} = \outcome = \outcomepair{\outcomein{}}{\outcomeout{}}$ with $b\notin \outcomein{}$
and $\outcome\succeq \DealDv{\GG}{v}$.
The only vertices that have such outcomes $\outcome$ are $t_i's$. For any $t_i$, 
the only outcome $\outcome$ that has these properties is $\outcomepair{s_i}{s_i}$.
But then $\HG$ would also have to contain $s_i$ contradicting the earlier statement.
(We note that
$\DealDv{\GG}{s_i} = \outcomepair{b}{b,\tildex_i}$
where $\tildex_i$ is the true literal of $x_i$,
and
there is no strictly better outcome for $s_i$.)

So we can assume from now on that $\HG$ contains $b$. The idea of the remaining argument is to show that
the assumption that $\HG$ dominates $\GG$ implies that in fact $\HG = \GG$ --- so
$\HG$ cannot strictly dominate $\GG$. To this end,
we examine the arcs of $\DG$ one by one. For each arc $(u,v)$, we use
the relation $\DealDv{\HG}{z} \succeq \DealDv{\GG}{z}$ for some $z\in\smbraced{u,v}$,
to show that $(u,v)$ belongs to $\HG$ if and only if
it belongs to $\GG$. We will divide this argument into a sequence of claims.


\myclaim{1}{$\HG$ contains all arcs $(v,b)$ and $(b,v)$, for $v\neq b$}. 
From the definition of the preference poset of $b$, 
$\HG$ must contain all incoming arcs of $b$. This gives us that $\HG$ is spanning.
For each $v \neq b$, any outcome $\outcome\succeq \DealDv{\GG}{v}$
that has outgoing arc $(v,b)$ must also have incoming arc $(b,v)$. This proves the claim.


\myclaim{2}{$\HG$ does not contain any arc $(a_j,c_j)$ or $(c_j,a_j)$}. 
Indeed,  no outcome of $a_j$ that has arc $(a_j,c_j)$ is better than
$\outcomepair{b}{b} = \DealDv{\GG}{a_j}$, so $\HG$ cannot contain $(a_j,c_j)$.
Similarly, no outcome of $c_j$ that has arc $(c_j,a_j)$ 
is better than $\outcomepair{b,T(c_j)}{b} = \DealDv{\GG}{c_j}$, so $\HG$ cannot contain $(c_j,a_j)$.


\myclaim{3}{For each literal $\tildex_i$ in clause $c_j$, $\HG$ contains $(\tildex_i,c_j)$ iff $\tildex_i$ is true}.
Suppose first that $\tildex_i$ is a literal in $c_j$ that is true. There is no outcome of $c_j$ that does not contain $(\tildex_i,c_j)$
and is better than $\outcomepair{b,T(c_j)}{b} = \DealDv{\GG}{c_j}$, so $\HG$ must contain $(\tildex_i,c_j)$.
Next, suppose that $\tildex_i$ is false. Then no outcome of  $\tildex_i$ that contains $(\tildex_i,c_j)$
is better than  $\outcomepair{b}{b} = \DealDv{\GG}{\tildex_i}$.  Thus $\HG$ cannot contain $(\tildex_i,c_j)$.


\myclaim{4}{For each literal $\tildex_i$, $\HG$ contains arc $(s_i,\tildex_i)$ iff $\tildex_i$ is true}.
Suppose first that literal $\tildex_i$ is true. 
From the previous claim, we have that the outgoing arcs to $C(\tildex_i)$ are in $\HG$.
There is no outcome of $\tildex_i$ that contains the arcs to $C(\tildex_i)$,
does not contain arc $(s_i,\tildex_i)$, and is better than
$\outcomepair{b,s_i}{b,C(\tildex_i)} = \DealDv{\GG}{\tildex_i}$. 
Thus  $\HG$ must contain $(s_i,\tildex_i)$.
Next, suppose that $\tildex_i$ is false, and let $\bar{\tildex}_i$ be the negation of $\tildex_i$ (that is, the true literal of $x_i$).
Then there is no outcome of $s_i$ that contains arc $(s_i,\tildex_i)$ and is better than
$\outcomepair{b}{b,\bar{\tildex}_i} = \DealDv{\GG}{s_i}$. Thus $\HG$ cannot contain $(s_i,\tildex_i)$.

\myclaim{5}{$\HG$ does not contain any arc $(s_i,t_i)$ or $(t_i,s_i)$}.
There is no outcome of $s_i$ that has an arc $(s_i,t_i)$ and is better than
$\outcomepair{b}{b,\tildex_i} = \DealDv{\GG}{s_i}$, where $\tildex_i$ is the true literal of $x_i$. 
So $\HG$ cannot contain $(s_i,t_i)$.
Also, there is no outcome of $t_i$ that has arc $(t_i,s_i)$, does not have arc $(s_i,t_i)$,
and is better than $\outcomepair{b}{b} = \DealDv{\GG}{t_i}$. So $\HG$ cannot contain $(t_i,s_i)$.


\smallskip

\noindent
$(\Leftarrow)$ 
Assume now that $\DG$ has a spanning subgraph $\GG$ that satisfies properties (c.1)-(c.3). From $\GG$ we will construct
a satisfying assignment for $\alpha$.

Since $\DealDv{\GG}{b} \succeq \DealDv{\DG}{b}$, $\GG$ must contain incoming arcs of $b$ from all 
other vertices. For each $v\neq b$, any outcome of $v$ that is at least as good as $\DealDv{\DG}{v}$ and contains
arc $(v,b)$ must also contain arc $(b,v)$. So $\GG$ contains all outgoing arcs of $b$.

Next, we claim that, for each variable $x_i$, $\GG$ contains at most one of arcs $(s_i,x_i)$ and $(s_i,\barx_i)$. 
Indeed, towards contradiction, suppose that $\GG$ contains both arcs $(s_i,x_i)$ and $(s_i,\barx_i)$. 
The best possible outcome of $s_i$ with both arcs $(s_i,x_i)$ and $(s_i,\barx_i)$ is  $\outcomepair{b,t_i}{x_i,\barx_i}$,
and using the preferences of $s_i$, we obtain
$\outcomepair{t_i}{t_i}\succ \outcomepair{b,t_i}{x_i,\barx_i} \succeq \DealDv{\GG}{s_i}$.
Regarding $t_i$, we have already established that $t_i$ has arcs to and from $b$, and the best such outcome for $t_i$
is $\outcomepair{b,s_i}{b}$. Thus, using the preferences of $t_i$ we obtain
$\outcomepair{s_i}{s_i}\succ \outcomepair{b,s_i}{b} \succeq    \DealDv{\GG}{t_i}$.
So we could take $\HG$ to consist of $s_i$, $t_i$, and arcs  $(s_i,t_i)$ and $(t_i,s_i)$, and 
this $\HG$ would strictly dominate $\GG$, contradicting our assumption that $\GG$ satisfies condition~(c.3).

Using the claim in the previous paragraph, we construct a satisfying assignment for $\alpha$ 
as follows: For each variable $x_i$, set it to true if $\GG$ contains $(s_i,x_i)$;
otherwise set it to false. (Note that $\GG$ may not contain any arc $(s_i,x_i)$, $(s_i,\barx_i)$, in which
case we could set the value of $x_i$ arbitrarily.)
This truth assignment is well defined. 

We now argue that this truth assignment satisfies $\alpha$. Consider any clause $c_j$.
Vertex $c_j$ must have at least one incoming arc $(\tildex_i,c_j)$ in $\GG$, because
otherwise we couldn't have $\DealDv{\GG}{c_j}\succeq \DealDv{\DG}{c_j}$. Similarly,
if $\GG$ contains this arc $(\tildex_i,c_j)$ then it must also contain arc $(s_i,\tildex_i)$,
because otherwise we couldn't have $\DealDv{\GG}{\tildex_i}\succeq \DealDv{\DG}{\tildex_i}$.
This implies that literal $\tildex_i$ is true in our truth assignment, so clause $c_j$
is true as well.
As this holds for each clause, we can conclude that $\alpha$ is satisfied.
\end{proof}

%% file: 06_sigma_two_main.tex
In the previous section, we showed that $\swapAtomic$ is $\NP$-hard.
In this section, we tighten the complexity classification of $\swapAtomic$, and show that it is in fact $\SigmaTwoP$-complete.
Recall that $\SigmaTwoP = \NP^{\NP}$ is the class of problems at the
2nd level of the polynomial hierarchy that consists of problems solvable non-deterministically in polynomial
time with an $\NP$ oracle.

Our proof is based on a reduction from a restricted variant of the $\ExistsForallDNF$ problem.
An instance of $\ExistsForallDNF$ is a boolean expression 
$\alpha = \exists \bfx \forall \bfy \beta(\bfx,\bfy)$, where 
$\bfx = (x_1,...,x_k)$ and $\bfy = (y_1,...,y_l)$ are vectors of boolean variables
and $\beta(\bfx,\bfy)$ is a quantifier-free boolean expression in
disjunctive normal form, that is $\beta(\bfx,\bfy) = \tau_1 \vee \tau_2 \vee ... \vee \tau_m$,
and each term $\tau_g$ is a conjunction of literals involving different variables. 
The goal is to determine whether $\alpha$ is true. $\ExistsForallDNF$ is a canonical $\SigmaTwoP$-complete 
problem~\cite{shafer_poly_hierarchy_compedium_2002,papadimitrou_book_1994}.
The restriction of $\ExistsForallDNF$ that we use in our proof, denoted $\ExistsForallDNFOneX$, consists of
instances $\alpha = \exists \bfx \forall \bfy \beta(\bfx,\bfy)$ where each term of
$\beta$ includes exactly one $\bfx$-literal and one or more $\bfy$-literals that involve different variables. 


\begin{lemma}\label{lem: sigma2dnf31 is sigma2 complete}
$\ExistsForallDNFOneX$ is $\SigmaTwoP$-complete.
\end{lemma}

The proof can be found in the supplemental material.


\input{sigma_two_figs-main}

We remember that 
$\swapAtomic$ is the decision problem of deciding whether a swap system has an atomic protocol.

\begin{theorem}
\label{thm:sigma2complete}
$\swapAtomic$ is $\SigmaTwoP$-complete.
\end{theorem}

The complete proof can be found in the supplemental material.
In the remainder of this section, 
we briefly present a high-level description of our reduction and the accompanying proof.
 
According to Theorem~\ref{thm:uni-and-snash}, a swap system $\swapSys = (\DG, \prefP)$
has an atomic swap protocol if and only if $\DG$ has a spanning subgraph $\GG$ with the following properties:
{(c.1)} $\GG$ is piece-wise strongly connected and has no isolated vertices,
{(c.2)} $\GG$ dominates $\DG$, and
{(c.3)} no subgraph $\HG$ of $\DG$ strictly dominates $\GG$.
This characterization is of the form $\exists \GG. \, \neg \exists \HG. \, \pi(\GG,\HG)$, 
where $\pi(\GG,\HG)$ is a polynomial-time decidable predicate, so it
immediately implies that $\swapAtomic$ is in $\SigmaTwoP$.
Thus it remains to show that $\swapAtomic$ is $\SigmaTwoP$-hard.

To prove $\SigmaTwoP$-hardness, 
we present a polynomial-time reduction from the above-defined decision problem
$\ExistsForallDNFOneX$. Let the given instance of $\ExistsForallDNFOneX$ be
$\alpha = \exists \bfx \forall \bfy \beta(\bfx,\bfy)$, where 
$\bfx = (x_1,...,x_k)$ and $\bfy = (y_1,...,y_l)$ are vectors of boolean variables
and $\beta(\bfx,\bfy) = \tau_1 \vee \tau_2 \vee ... \vee \tau_m$,
where each $\tau_g$ is a conjunction of one $\bfx$-literal and one or more $\bfy$-literals. 
Our reduction converts $\alpha$ into a swap system $\swapSys = (\DG, \prefP)$
such that $\alpha$ is true if and only if  $\DG$ has a spanning subgraph $\GG$
that satisfies conditions (c.1)-(c.3) from Theorem~\ref{thm:uni-and-snash}.

The following informal interpretation of $\ExistsForallDNFOneX$ will be helpful in understanding our reduction.
Say that a truth assignment to some variables ``kills'' a term $\tau_g$ if it sets one of its literals to false.
A truth assignment $\bfphi$ to the $\bfx$-variables will kill some terms, while others will survive.
Thus $\alpha$ will be true for assignment $\bfphi$ iff there is no assignment $\bfpsi$ for the $\bfy$-variables that kills all
terms that survived $\bfphi$.
In our reduction, the existence of this assignment $\bfphi$ will be represented by the existence of subgraph $\GG$.
The non-existence of $\bfpsi$ that kills all terms that survived $\bfphi$ will be represented by the non-existence of
a subgraph $\HG$ that strictly dominates $\GG$.

Throughout this section, the negation of a boolean variable $x_i$ will be denoted $\barx_i$.
We will also use notation $\tildex_i$ for an unspecified literal of $x_i$, that is $\tildex_i \in \braced{x_i,\barx_i}$. 
The same conventions apply to the variables $y_j$.

We now give an overview of our reduction. The digraph $\DG$ consists of several ``gadgets''.
There will be $\exists$-gadgets, which correspond to the variables $x_i$ and will be used to set their values, 
through the choice of subgraphs that $\GG$ includes.
Then there is the $\forall$-gadget, that contains ``sub-gadgets'' representing the literals $\tildey_j$ and the terms $\tau_g$. 
These gadgets will allow for the values of the variables $y_j$ 
to admit all possible assignments.
If any setting of these values kills all terms not yet killed by the variables $x_i$, this gadget will contain a subgraph
$\HG$ that strictly dominates $\GG$. 
Figure~\ref{fig: main-sigma2 reduction 1} shows a single $\exists$-gadget and 
Figure~\ref{fig: main-sigma2 reduction 2} shows the $\forall$-gadget.
As we explain the high-level intuition,
we gradually visit vertices and explain their purpose.

The argument is based on several ideas. 
One, 
we design the preference posets of $\tildex_i$'s so that
$\GG$ is forced to choose between two possible subsets of arcs within the $\exists$-gadget.
The choice between these two subsets of arcs corresponds to choosing a truth assignment for variable $x_i$.
We focus on the literals $\tildex_i$ that are set to false, since these kill the terms where they appear.
If $\tildex_i$ is set to false, its arcs to the terms $\tau_g$'s in which the literal appears will be included in $\GG$ (the first subset), otherwise its arc to $\tildez_i$ will be included in $\GG$ (the second subset).

Another idea is that vertices outside of the $\forall$-gadget have their preference posets defined in such a way that their arcs 
in $\GG$ define an outcome that is already the best for them.
Therefore, if a subgraph $\HG$ that strictly dominates $\GG$ does indeed exist, we know it must appear in the $\forall$-gadget.
This leads into the key idea of the $\forall$-gadget.
The vertices in this gadget 
can have
outcomes that are better than their outcomes in $\GG$.
All the arcs in these better outcomes together form the cycle
\begin{equation}
    \begin{array}{lcl}
	\calC &=&	
	q_0 \to \tildey_1 \to ... \to \tildey_l\to q_l \to 
	\\ &&
	p_0 \to \tau_1 \to ... \tau_m \to p_m \to 
	\\ &&
	q_0
	\end{array}
	\label{eqn: sigma2-complete cycle main}
\end{equation}
for some choice of the literals $\tildey_1, ..., \tildey_l$.
We design the preference posets of each $\tau_g$ so that its outcome in $\GG$ can only be improved (specifically, towards $\calC$) only if it receives an arc from one of its literals --- in other words, if it is killed by that literal.
This way, $\GG$ will have a strictly dominating subgraph $\HG$ (namely cycle $\calC$) only if all terms are killed, i.e. when $\alpha$ is false.

Next, we provide brief insight to the important vertices and how they help capture the ideas above.
Firstly, we want to simulate a truth assignment for variable $x_i$, which we represent by having $\GG$ choose between two subsets of arcs in the corresponding $\exists$-gadget for $x_i$.
Intuitively, one subset corresponds to assigning $x_i$ to true while the other subset corresponds to assigning $x_i$ to false.
These two subsets of arcs are established by how we define the preference posets of $x_i$ and $\barx_i$.
In order to force $\GG$ to make a choice (instead of taking all the arcs), we introduce the auxiliary vertex $a$, which has arcs to every literal $\tildex_i$.
%
%
The graph $\DG$ has two strongly connected components:
(1) $a$ and $a'$, and (2) all the other vertices.
We claim graph $\GG$ cannot include any arcs from $a$ to the literal vertices.
Otherwise, since there is no edge from the second component to the first,
dropping those arcs always results in a better outcome for the first component.
However, that contradicts condition (c.3).
%
%
%
Then, for $\GG$ to satisfy condition~(c.2), $\GG$ is forced to make a choice between the two subsets of arcs.
Specifically, $\GG$ must choose to include all arcs from $x_i$ to its terms (corresponds to setting $x_i$ to be false) or all arcs from $\barx_i$ to its terms (corresponds to setting $x_i$ to be true).

Next, we want to simulate a term $\tau_g$ being killed.
We achieve this by designing the preference poset of each $\tau_g$ so that if it receives its arc from its $\bfx$-literal,
it would prefer its outcome in the cycle $\calC$ over its outcome in $\GG$.
A term $\tau_g$ can also be killed by one of its $\bfy$-literals, which we describe later.

Now, we want to simulate checking whether there is a truth assignment for the $\bfy$-variables that make $\forall \bfy \beta (\bfphi, \bfy)$ false, where $\bfphi$ is a truth assignment over the $\bfx$ variables.
In other words, $\GG$'s assignment of the $\bfx$ variables have killed some terms and now we want to see if $\HG$ can give an assignment of the $\bfy$ variables that kill the surviving terms.

First, we need to simulate a truth assignment for each variable $y_j$.
This is simple: we flank $y_j$ and $\bary_j$ by vertices $q_{j-1}$ and $q_j$ as seen in Figure~\ref{fig: main-sigma2 reduction 2}.
We define the preference posets of $q_{j-1}$ and $q_j$ in a way that only one of $y_j$ or $\bary_j$ can be included in the cycle $\calC$, 
thus forcing $\HG$ to choose between them.
If $\HG$ selects a vertex $\tildey_j$, then $\tildey_j$ will send an arc to every term it appears in.
This corresponds to assigning $\tildey_j$ to false.

We additionally define $\tau_g$'s poset so that if it receives an arc from any of its $\bfy$-literals, it wants to join $\calC$.
At this point, we have represented $\tau_g$ being killed if it receives either its $\bfx$-literal arc or any of its $\bfy$-literal arcs.
(The $\tildez_i$ vertices are actually used for this purpose. 
They help distinguish between when a term is killed by their $\bfx$-literal and when a term survived in which it needs to be killed by a $\bfy$-literal.)

The $\forall$-gadget is designed in the following manner:
the preference posets of the $q_j$, $\tildey_j$, and $p_g$ vertices are such that if every $\tau_g$ prefers $\calC$, so will they;
otherwise, if any $\tau_g$ does not prefer $\calC$, then none of the vertices can cooperatively deviate to improve their outcomes.
In other words, each $\tau_g$ acts as a bottleneck for the cycle $\calC$, thus we focus only on the $\tau_g$'s.

We give an analogy to better understand the remainder of the reduction.
Each $\tau_g$ is given a vote to whether or not they want to participate in cycle $\calC$.
In order for $\calC$ to pass, it must receive a \emph{unanimous} vote from every $\tau_g$.
Vertex $\tau_g$ only casts its vote to join $\calC$ if it receives an arc from either its $\bfx$-literal arc or any of its $\bfy$-literal arcs (corresponds to being killed).
Then, $\GG$'s selection in each $\exists$-gadget (truth assignment over $\bfx$ variables) caused \emph{some} $\tau_g$'s to vote for $\calC$.
Now, $\HG$ is tasked with selecting vertices $\tildey_1, ..., \tildey_l$ (truth assignment over $\bfy$ variables) so that the \emph{remaining} $\tau_g$'s also vote for $\calC$.
At the end of this, if the $\tau_g$'s unanimously voted for $\calC$, then there is an $\HG$ that strictly dominates $\GG$, namely the subgraph induced by $\calC$.
(This corresponds to giving an assignment $\bfy \mapsto \bfpsi$ such that $\beta (\bfphi, \bfpsi)$ is false.)
Otherwise, if $\HG$ cannot give such a selection over $\tildey_1, ..., \tildey_l$, then there is no $\HG$ that strictly dominates $\GG$.
(This corresponds to $\forall \bfy \beta(\bfphi, \bfy)$ being true, i.e. $\alpha$ is true.)

The remaining vertices
are primarily used for convenience and to influence the behavior/preference posets of their neighbors.
In other words, they make the topology of $\GG$ and $\HG$ predictable, holding them to a particular form.
For example, vertex $b$ is used to guarantee condition~(c.1), the piece-wise strong connectivity condition.
Also, it is used where vertices would otherwise have no incoming or outgoing arcs.

To conclude the section, we provide some brief insight to both directions of the proof. 
In the $(\Rightarrow)$ implication, 
we show that if $\alpha$ is true, 
then the swap graph $\DG$ has a subgraph $\GG$ that satisfies the properties of Theorem~\ref{thm:uni-and-snash}.
We begin by fixing some truth assignment $\bfx \mapsto \bfphi$ that makes $\alpha$ true.
We convert $\bfphi$ into a graph $\GG$ that satisfies the properties of Theorem~\ref{thm:uni-and-snash} using the ideas described above.
Conditions~(c.1) and~(c.2) can be verified by routine inspection, leaving condition~(c.3) that $\GG$ does not have a strictly dominating subgraph $\HG$.
The idea is, towards contradiction, if such an $\HG$ existed, we could convert it into an assignment $\bfpsi$ of the $\bfy$ variables so that $\beta(\bfphi, \bfpsi)$ is false.
This contradicts the fact that $\alpha$ is true.

In the $(\Leftarrow)$ implication,
we show that if $\DG$ has a subgraph $\GG$ that satisfies the properties of theorem~\ref{thm:uni-and-snash}, then $\alpha$ is true.
We begin by showing that the topology of $\GG$ must have a certain form; specifically, it is representative of the graph $\GG$ we constructed in the proof of the $(\Rightarrow)$ implication.
This allows us to reconstruct an assignment of the $\bfx$ variables.
We then show, again by contradiction, that $\forall \bfy \beta(\bfphi, \bfy)$ must be true.
If it were not, we can take a falsifying assignment $\bfy \mapsto \bfpsi$ and convert it into a subgraph $\HG$ that strictly dominates $\GG$.
However, this contradicts condition~(c.3) of $\GG$.

%

%% file: sigma_two_figs-main.tex

\begin{figure*}[ht]
\centering
\includegraphics[width=4.75in]{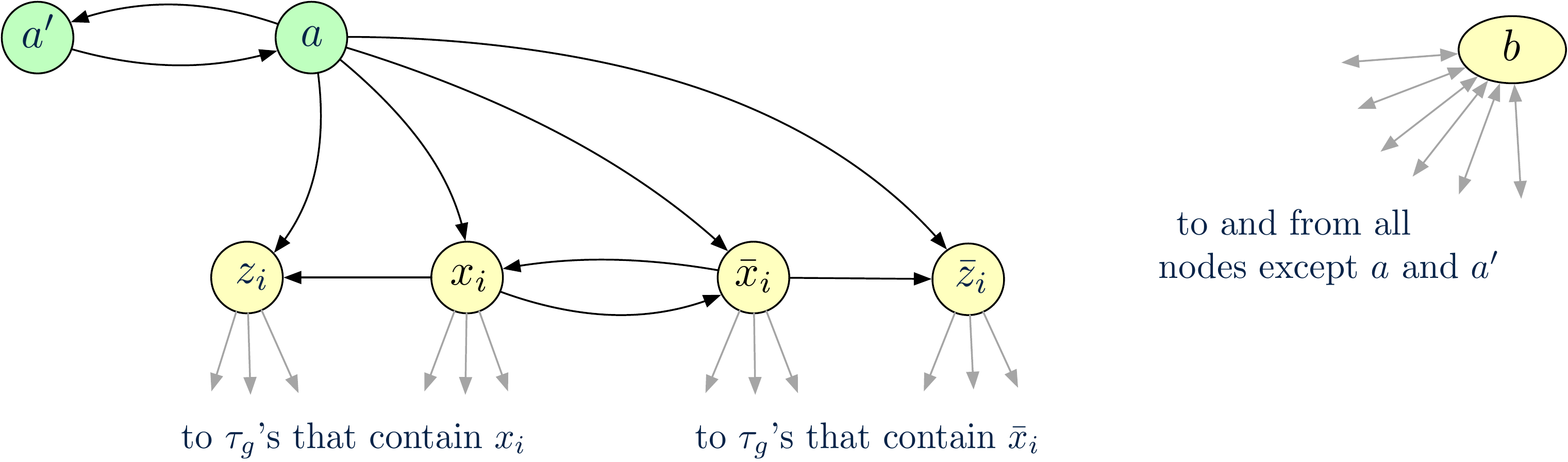} 
\caption{The construction of digraph $\DG$ in the proof of $\SigmaTwoP$-hardness. This figure shows vertices $a$, $a'$, $b$, and
an $\exists$-gadget for variable $x_i$.
The arcs to and from $b$ are shown as bi-directional arrows at $b$.}
\label{fig: main-sigma2 reduction 1}
\end{figure*}

\begin{figure*}[ht]
\centering
\includegraphics[width=6in]{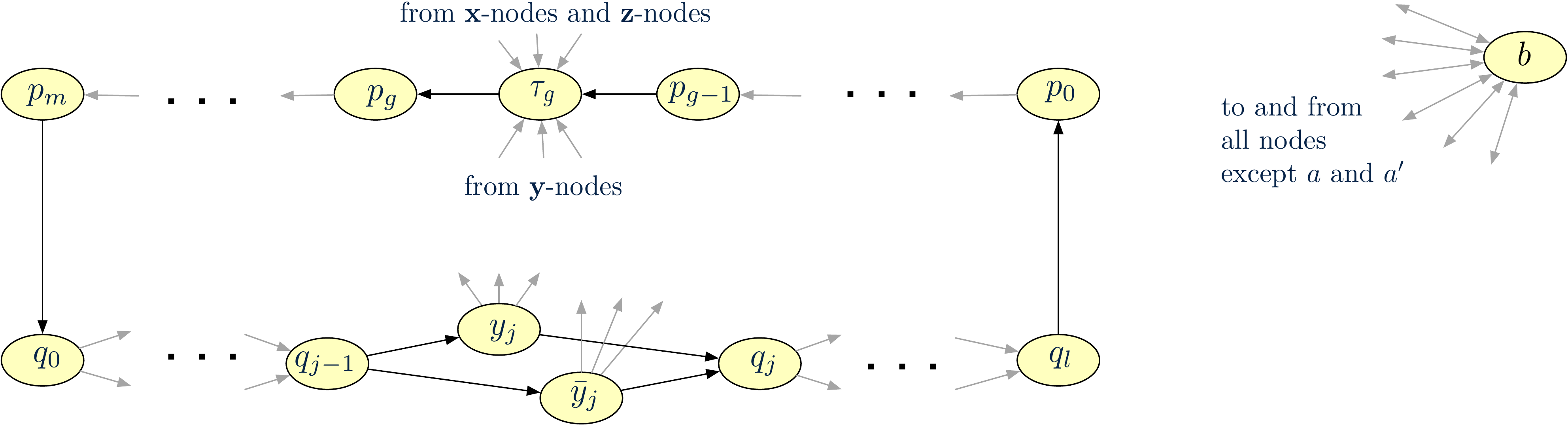} 
\caption{The construction of digraph $\DG$ in the proof of $\SigmaTwoP$-hardness. This figures shows the
$\forall$-gadget, namely the
part of $\DG$ that contains the vertices that simulate setting the values of the $y_j$-variables and the terms $\tau_g$.
The arcs to and from $b$ are shown as bi-directional arrows at $b$.}
\label{fig: main-sigma2 reduction 2}
\end{figure*}


%% file: 07_related.tex
The fair exchange problem 
\cite{micali03, franklin98, probfairexchange, asokan97, asokan1997}
was of interest even before the blockchain technology.  
It arises when two parties want to exchange their assets,
and the outcome  must be either 
that the two parties end up trading their assets, 
or that they both keep their assets.
However, in contrast to the swap problem, some trust in a third party is often assumed.
The optimistic fair exchange protocol~\cite{micali03}
relies on invisible trusted parties: parties that work as a background service and
intervene only in case of a misbehaviour. 
Similarly, the secure group barter protocol \cite{franklin98} studies multi-party barter 
with semi-trusted agents. 

To the best of our knowledge, it was back in 2013 when 
the notion of cross-chain swaps first emerged in an online forum \cite{atomictransfertalk}.
Atomic cross-chain swap is since an active problem for the blockchain community \cite{btcwiki,atomictransfertalk, htlc, bitcoinhtlc}.
The two wiki pages \cite{btcwiki} and \cite{atomictransfertalk} and later
platforms such as  deCRED \cite{decred} proposed protocols for bilateral swaps.
However, these projects offer only two-party transactions.
Later, protocols for cross-chain swaps and transactions \cite{Herlihy18,herlihyDeals,heilman2020arwen,thyagarajan2022universal}
emerged that can work for an arbitrary number of parties; however, they assumed the predefined preference relation that we saw earlier
for all the parties.

These protocols motivated a host of follow-up research.
The time and space complexity \cite{imoto2019atomic} and privacy guarantees \cite{deshpande2020privacy} of the protocol were improved.
The former \cite{imoto2019atomic} uses a model where each asset is assigned two numerical values, one by its current owner and one
by the intended recipient. These values can then be used to determine preferences for each party, and can be 
extended to sets of parties by considering the difference between the total values of incoming and outgoing assets.
To assure that their swap graphs have atomic protocols, restrictions (similar in spirit to
our Theorem~\ref{thm:uni-and-snash}) are placed on allowed swap graphs. 
As we discussed in the introduction, such value-based preferences cannot express dependencies between assets.
So the model in \cite{imoto2019atomic} would not capture some natural scenarios, for example trades involving assets
from an investment portfolio with fixed proportions between different asset classes.
Their way of extending individual preferences to coalitions (sets of parties) is different from our model, and
it involves a tacit assumption that the coalition members agree on these values.
Nevertheless, the approach in~\cite{imoto2019atomic} is natural and worth studing, and in particular it would
be of interest to investigate the time complexity to determine whether a swap graph has an atomic protocol in that model. 
We suspect that this problem may be computationally easier than for our swap systems.

Further, extensions to  support off-chain steps \cite{shadabTransactions} 
and reduce the asset lock-up time \cite{Xue21Hedging} appeared.
Others presented hardness and impossibility results \cite{zamyatin2021sok,chan2021brief}
formal verification \cite{nehai2022tla}, and  protocols with 
all-or-nothing guarantees \cite{zakhary13atomic} and
success guarantees under synchrony assumptions \cite{van2020feasibility}.
Others proposed moving assets \cite{sigwart2020decentralized} and smart contracts \cite{fynn2020smart} across blockchains,
and executing code that spans multiple blockchains 
\cite{GeneralRobinson21},
and
presented implementations for industrial blockchains \cite{wehatomicswapweb,wehatomicswapweb2,bu2020cross,thyagarajan2021universal}.

Payment channel networks 
process multi-hop payments 
in the same blockchain
through a sequence of channels 
using Hash Timelock Contracts 
\cite{poon2016bitcoin,reidenweb}
or 
adaptor signatures \cite{malavolta2018anonymous}.
Recent protocols such as AMCU \cite{egger2019atomic}, Sprites \cite{miller2019sprites} and Thora \cite{aumayr2022thora} support more general topologies for transactions.

In contrast to previous work,
this paper presented a generalized model of swaps where each party can specify a personalized preference on their set of incoming and outgoing assets in a finer manner, e.g. dependencies between subsets of acquired and traded assets.

%% file: 11_sigma2_completeness.tex

In this section, we give the complete, detailed proof described in \ref{sec:sigma2 completeness-main}.
That is, we consider the complexity of determining whether a swap system has an atomic swap protocol, showing that this problem is $\SigmaTwoP$-complete. 
Recall that $\SigmaTwoP = \NP^{\NP}$ is the class of problems at the 2nd level of the polynomial hierarchy that consists of problems solvable non-deterministically in polynomial time with an $\NP$ oracle.

Our proof is based on a reduction from a restricted variant of the $\ExistsForallDNF$ problem.
An instance of $\ExistsForallDNF$ is a boolean expression 
$\alpha = \exists \bfx \forall \bfy \beta(\bfx,\bfy)$, where 
$\bfx = (x_1,...,x_k)$ and $\bfy = (y_1,...,y_l)$ are vectors of boolean variables
and $\beta(\bfx,\bfy)$ is a quantifier-free boolean expression in
disjunctive normal form, that is $\beta(\bfx,\bfy) = \tau_1 \vee \tau_2 \vee ... \vee \tau_m$,
and each term $\tau_g$ is a conjunction of literals involving different variables. 
The goal is to determine whether $\alpha$ is true. $\ExistsForallDNF$ is a canonical $\SigmaTwoP$-complete 
problem~\cite{shafer_poly_hierarchy_compedium_2002,papadimitrou_book_1994}\footnote{%
Notations for this problem and its variants vary across the literature. Our notations use the convention in~\cite{shafer_poly_hierarchy_compedium_2002}.
}.
The problem remains $\SigmaTwoP$-complete even restricted to instances
where each term in $\beta$ has only three literals. We denote this variant by $\ExistsForallThreeDNF$.

Throughout this section, the negation of a boolean variable $x_i$ will be denoted $\barx_i$.
We will also use notation $\tildex_i$ for an unspecified literal of $x_i$, that is $\tildex_i \in \braced{x_i,\barx_i}$. 
The same conventions apply to the variables $y_j$.

The restriction of $\ExistsForallDNF$ that we use in our proof, denoted $\ExistsForallDNFOneX$, consists of
instances $\alpha = \exists \bfx \forall \bfy \beta(\bfx,\bfy)$ where each term of
$\beta$ includes exactly one $\bfx$-literal and one or more $\bfy$-literals that ifnvolve different variables.

We first prove the following lemma:


\begin{lemma}\label{lem:sigma2-app}
$\ExistsForallDNFOneX$ is $\SigmaTwoP$-complete.
\end{lemma}

\begin{proof}
We show how to convert a given instance $\alpha = \exists \bfx \forall \bfy \beta(\bfx,\bfy)$
of $\ExistsForallThreeDNF$ into an instance $\alpha'$ of $\ExistsForallDNFOneX$ such that $\alpha$ is true iff $\alpha'$ is true.

First, we can assume that  $\beta$
does not have terms with only $\bfx$-literals, since such formulas $\alpha$ are trivially true.
All terms that have exactly one $\bfx$-literal will remain unchanged.

Consider a term with two $\bfx$-literals, say $\tau_g = \tildex_p\wedge \tildex_q \wedge \tildey_r$. Add another variable $y'$ and
replace $\tau_g$ by $(\tildex_p \wedge \tildey_r \wedge y') \vee (\bary' \wedge \tildex_q \wedge \tildey_r)$. Let $\beta'$ be the boolean expression obtained from $\beta$ by
this replacement,
and $\alpha' = \exists \bfx \forall \bfy \forall y' \beta'(\bfx,\bfy,y')$. Then, by straightforward
verification, $\alpha$ is true for a given truth assignment for $\bfx$ if and only if $\alpha'$ is true for the
same assignment for $\bfx$.

By applying these replacements, we will eventually eliminate all terms that
have two or zero $\bfx$-literals, thus converting $\alpha$ into the $\ExistsForallDNFOneX$ form.
\end{proof}


\begin{theorem}
\label{thm:sigma2complete-app}
Let $\swapAtomic$ be the decision problem of deciding whether a swap system has an atomic protocol.
$\swapAtomic$ is $\SigmaTwoP$-complete.
\end{theorem}

\begin{proof}
According to Theorem~\ref{thm:uni-and-snash}, a swap system $\swapSys = (\DG, \prefP)$
has an atomic swap protocol if and only if $\DG$ has a spanning subgraph $\GG$ with the following properties:
{(c.1)} $\GG$ is piece-wise strongly connected and has no isolated vertices,
{(c.2)} $\GG$ dominates $\DG$, and
{(c.3)} no subgraph $\HG$ of $\DG$ strictly dominates $\GG$.
This characterization is of the form $\exists \GG \, (\,\neg \exists \HG : \pi(\GG,\HG) \,)$, 
where $\pi(\GG,\HG)$ is a polynomial-time decidable predicate, so it
immediately implies that $\swapAtomic$ is in $\SigmaTwoP$.
Thus it remains to show that $\swapAtomic$ is $\SigmaTwoP$-hard.

To prove $\SigmaTwoP$-hardness, we give a polynomial-time reduction from the above-defined decision problem
$\ExistsForallDNFOneX$. Let the given instance of $\ExistsForallDNFOneX$ be
$\alpha = \exists \bfx \forall \bfy \beta(\bfx,\bfy)$, where 
$\bfx = (x_1,...,x_k)$ and $\bfy = (y_1,...,y_l)$ are vectors of boolean variables
and $\beta(\bfx,\bfy) = \tau_1 \vee \tau_2 \vee ... \vee \tau_m$,
with each $\tau_g$ being a conjunction of one $\bfx$-literal and one or more $\bfy$-literals. 
Our reduction converts $\alpha$ into a swap system $\swapSys = (\DG, \prefP)$
    such that $\alpha$ is true if and only if  $\DG$ has a spanning subgraph $\GG$
that satisfies conditions (c.1)-(c.3) from Theorem~\ref{thm:uni-and-snash}.

The following informal interpretation of $\ExistsForallDNFOneX$ will be helpful in understanding our reduction.
Say that a truth assignment to some variables \emph{kills} a term $\tau_g$ if it sets one of its literals to false.
A truth assignment $\bfphi$ to the $\bfx$-variables will kill some terms, while other will survive.
Thus $\alpha$ will be true for assigment $\bfphi$ iff there is no assignment $\bfpsi$ for the $\bfy$-variables that kills all
terms that survived $\bfphi$.
In our reduction, the existence of this assigment $\bfphi$ will be represented by the existence of subgraph $\GG$.
The non-existence of $\bfpsi$ that kills all terms that survived $\bfphi$ will be represented by the non-existence of
a subgraph $\HG$ that strictly dominates $\GG$.

\smallskip

We now describe our reduction. The digraph $\DG$ will consists of several ``gadgets''.
There will be $\exists$-gadgets, which correspond to the variables $x_i$ and will be used to set their values, through an appropriate
choices of subgraph $\GG$.
Then there is the $\forall$-gadget, that contains ``sub-gadgets'' representing the literals $\tildey_j$ and the terms $\tau_g$. These gadgets
will allow for the values of the variables $y_j$ to be set in all possible ways.
If any setting of these values kills all terms not yet killed by the variables $x_i$, this gadget will contain a subgraph
$\HG$ that strictly dominates $\GG$. 

In addition to these gadgets, digraph $\DG$ has three auxiliary vertices $a$, $a'$ and $b$. Vertices $a$ and $a'$ are
connected by arcs $(a,a')$ and $(a',a)$. Vertex $a$ also has some outgoing arcs that will be described later.
Vertex $b$ is connected by arcs to and from all other vertices of $\DG$ except $a$ and $a'$.

Next, we describe the gadgets (for now, we specify only their vertices and arcs --- the preference posets will be defined later). 
The $\exists$-gadget corresponding to $x_i$ is shown in Figure~\ref{fig: sigma2 reduction 1}.
It's constructed as follows:
{
\renewcommand\labelitemi{---}
\begin{itemize}[leftmargin=*]
\item 
For $i=1,...,k$, create vertices $x_i$, $\barx_i$, $z_i$ and $\barz_i$, with arcs $(a,x_i)$, $(a,\barx_i)$,
$(a,z_i)$, $(a,\barz_i)$, $(x_i,\barx_i)$, $(\barx_i,x_i)$, $(x_i,z_i)$, and $(\barx_i,\barz_i)$. 
Throughout the proof we will use notation $\tildez_i$ for the vertex corresponding to $\tildex_i$,
that is $\tildez_i = z_i$ if $\tildex_i = x_i$, and $\tildez_i = \barz_i$ if $\tildex_i = \barx_i$.
\end{itemize}
}

\input{sigma_two_figs}

The $\forall$-gadget is shown in Figure~\ref{fig: sigma2 reduction 2}. It's constructed as follows:
{
\renewcommand\labelitemi{---}
\begin{itemize}[leftmargin=*]
\item 
For $j =0,...,l$, create vertices $q_j$. For $j = 1,...,l$, create vertices $y_j$ and $\bary_j$ and 
arcs $(q_{j-1},y_j)$, $(q_{j-1},\bary_j)$, $(y_j, q_j)$, and $(\bary_j,q_j)$.
\item
For $g = 0,...,m$, create vertices $p_g$. For $g = 1,...,m$,
create vertices $\tau_g$ and arcs $(p_{g-1},\tau_g)$ and $(\tau_g,p_g)$.
\item 
Create arcs $(q_l,p_0)$ and $(p_m,q_0)$.
\item
For each $g = 1,...,m$, and for each literal $\tildey_j$ in $\tau_g$, create arc $(\tildey_j,\tau_g)$.
\end{itemize}
}

To complete the construction of $\DG$, we add arcs between $\exists$-gadgets and the $\forall$-gadget:

{
\renewcommand\labelitemi{---}
\begin{itemize}[leftmargin=*]
\item
For each $g = 1,...,m$, if $\tildex_i$ is the $\bfx$-literal in $\tau_g$ (there is exactly one, by the
definition of $\ExistsForallDNFOneX$), create arcs $(\tildex_i,\tau_g)$ and $(\tildez_i,\tau_g)$.
\end{itemize}
}



Next, we need to define preference posets for all vertices. As explained in Section~\ref{sec:swap systems},
all preference posets are specified by their list of generators. An outcome $\outcomepair{\outcomein{}}{\outcomeout{}}$
of each vertex $v$ is specified by lists $\outcomein{}$ and $\outcomeout{}$ of its in-neighbors and out-neighbors, respectively. 
With this convention, the generators of all preference posets are:

{
\renewcommand\labelitemi{---}
\begin{itemize}[leftmargin=*]
\item
Vertices $a$, $a'$, and $b$ do not have any generators.
\item 
The generators for the $\exists$-gadget corresponding to variable $x_i$ are as follows.
For each literal $\tildex_i$, its generators are
$\Dealv{\tildex_i} \prec \outcomepair{b}{b,{\bartildex}_i, T(\tildex_i)}$ and $\Dealv{\tildex_i} \prec \outcomepair{b,\bartildex_i}{b, \tildez_i}$,
where ${\bartildex}_i$ is the negation of $\tildex_i$ and $T(\tildex_i)$ is the set of terms that contain literal $\tildex_i$.
The generators of $\tildez_i$ are $\Dealv{\tildez_i} \prec \outcomepair{b}{b}$ and $\DealDv{}{\tildez} \prec \outcomepair{b, \tildex_i}{b, T(\tildex_i)}$.
\item
For each literal $\tildey_j$, its generators are
$\Dealv{\tildey_j} \prec \outcomepair{b}{b}$ and
$\outcomepair{b}{b} \prec \outcomepair{q_{j-1}}{q_{j},T(\tildey_j)}$.
The generators of  $q_j$, where $j\notin\braced{0,l}$, are $\Dealv{q_j} \prec \outcomepair{b}{b}$ and
$\outcomepair{b}{b} \prec \outcomepair{\tildey_j}{\tildey_{j+1}}$, for all literals 
$\tildey_j \in \braced{y_j,\bary_j}$ and $\tildey_{j+1} \in  \braced{y_{j+1},\bary_{j+1}}$.
\item
The generators of $q_0$ are $\Dealv{q_0} \prec \outcomepair{b}{b}$ and 
$\outcomepair{b}{b} \prec \outcomepair{p_m}{\tildey_{1}}$, for all $\tildey_1 \in \braced{y_1,\bary_1}$.
The generators of $q_l$ are $\Dealv{q_l} \prec \outcomepair{b}{b}$ and 
$\outcomepair{b}{b} \prec \outcomepair{\tildey_{l}}{p_0}$, for all $\tildey_l \in \braced{y_l,\bary_l}$.
\item
For each term $\tau_g$, letting $\tildex_i$ be the unique $\bfx$-literal in $\tau_g$,
its generators are:
$\Dealv{\tau_g} \prec \outcomepair{b,\tildex_i}{b}$,
$\outcomepair{b,\tildex_i}{b} \prec \outcomepair{p_{g-1},L}{p_g}$ 
for any subset $L$ of the $\bfy$-literals in $\tau_g$,
$\Dealv{\tau_g} \prec \outcomepair{b,\tildez_i}{b}$,
and
$\outcomepair{b,\tildez_i}{b} \prec \outcomepair{p_{g-1},L'}{p_g}$ 
for any \emph{non-empty} subset $L'$ of the $\bfy$-literals in $\tau_g$.
For each $p_g$, where $g\notin\braced{0,m}$, its generators are $\Dealv{p_g} \prec \outcomepair{b}{b}$ and
$\outcomepair{b}{b} \prec \outcomepair{\tau_g}{\tau_{g+1}}$.
\item
The generators of $p_0$ are $\Dealv{p_0} \prec \outcomepair{b}{b}$ and 
$\outcomepair{b}{b} \prec \outcomepair{q_l}{\tau_{1}}$.
The generators of $p_m$ are $\Dealv{p_m} \prec \outcomepair{b}{b}$ and 
$\outcomepair{b}{b} \prec \outcomepair{\tau_{m}}{q_0}$.
\end{itemize}
}

With this, the description of $\swapSys$ is complete. The construction of $\swapSys$ clearly takes time
that is polynomial in the size of $\alpha$. Applying Theorem~\ref{thm:uni-and-snash},  it remains to show that
$\alpha$ is true if and only if $\DG$ has a spanning subgraph $\GG$ with properties (c.1)-(c.3).

The argument is based on several ideas. 
One, We design the preference posets of $\tildex_i$'s so that $\GG$ is forced to choose between two possible subsets of arcs within the $\exists$-gadget.
The choice between these two subsets of arcs corresponds to choosing a truth assignment for variable $x_i$.
We focus on the literals $\tildex_i$ that are set to false, since these kill the terms where they appear.
If $\tildex_i$ is set to false, its arcs to the terms $\tau_g$'s in which the literal appears will be included in $\GG$ (the first subset), otherwise its arc to $\tildez_i$ will be included in $\GG$ (the second subset).

Another idea is that vertices outside of the $\forall$-gadget have their preference posets defined in such a way that their arcs
in $\GG$ define an outcome that is already the best for them.
Therefore, if a subgraph $\HG$ that strictly dominates $\GG$ does indeed exist, we know it must appear in the $\forall$-gadget.
This leads into the key idea of the $\forall$-gadget.
The vertices in this gadget can have outcomes that are better than their outcomes in $\GG$.
All the arcs in these better outcomes together form the cycle
\begin{equation}
    \begin{array}{lcl}
	\calC &=&	
	q_0 \to \tildey_1 \to ... \to \tildey_l\to q_l \to 
	\\ &&
	p_0 \to \tau_1 \to ... \tau_m \to p_m \to q_0
	\end{array}
	\label{eqn: sigma2-complete cycle}
\end{equation}
for some choice of the literals $\tildey_1, ..., \tildey_l$.
We design the preference posets of each $\tau_g$ so that its outcome in $\GG$ can only be improved (specifically, towards $\calC$) only if it receives an arc from one of its literals --- in other words, if it is killed by that literal.
This way, $\GG$ will have a strictly dominating subgraph $\HG$ (namely cycle $\calC$) only if all terms are killed, i.e. when $\alpha$ is false. 
The formal proof follows.

\smallskip
$(\Rightarrow)$
Suppose $\alpha$ is true. Fix some truth assignments $\bfx\mapsto \bfphi$ for which 
$\forall \bfy \beta(\bfphi,\bfy)$ is true. This means that for each truth assignment $\bfy\mapsto \bfpsi$
the boolean expression $\beta(\bfphi,\bfpsi)$ is true. For each truth assignment $\bfy\mapsto \bfpsi$ we can thus
choose an index ${h(\bfpsi)}$ for which term $\tau_{h(\bfpsi)}$ is true.

Using this assignment $\bfx\mapsto \bfphi$, we construct a spanning subgraph $\GG$ of $\DG$ that satisfies the three conditions (c.1)-(c.3). 
$\GG$ will contain all vertices from the above construction and all arcs that connect
$b$ to all other vertices except $a$ and $a'$, in both directions. Vertices $a$ and $a'$ will be connected by arcs $(a,a')$ and $(a',a)$.
This makes $\GG$ spanning and piece-wise strongly connected, with one strongly connected component consisting of vertices $a$ and $a'$
and the other consisting of all other vertices. So (c.1) holds. 

Next, we define the arcs of $\GG$ for the vertices in the $\exists$-gadgets. For any given $i$,
if $\phi(x_i) = 1$, add to $\GG$ the following arcs:
$(x_i,z_i)$, $(\barx_i,x_i)$, all arcs $(z_i,\tau_j)$ for terms $\tau_j\in T(x_i)$,
and all arcs $(\barx_i,\tau_j)$ for terms $\tau_j \in T(\barx_i)$.
Symmetrically, if $\phi(x_i) = 0$, add to $\GG$ the following arcs:
$(\barx_i,\barz_i)$, $(x_i,\barx_i)$, all arcs $(\barz_i,\tau_j)$ for terms $\tau_j\in T(\barx_i)$,
and all arcs $(x_i,\tau_j)$ for terms $\tau_j \in T(x_i)$.
(Note that we add the arcs from false literals to the terms that they kill, and from true literals to the corresponding nodes $\tildez_i$.)
We now need to verify conditions~(c.2) and~(c.3).

\smallskip

Condition~(c.2) can be verified by routine inspection of all nodes. For each vertex $v$ we need to check that
$\DealDv{\GG}{v}\succeq \DealDv{\DG}{v}$.
For $v \in \braced{a',b}$, we have $\DealDv{\GG}{v} = \DealDv{\DG}{v}$.
For $v =a$, $\DealDv{\GG}{a} = \outcomepair{a'}{a'} \succ \DealDv{\DG}{v}$.
For $v = \tildex_i$ there are two cases: either $\DealDv{\GG}{\tildex_i} = \outcomepair{b,\bartildex_i}{b,\tildez_i}$
(if $\phi(\tildex_i) = 1$) or $\DealDv{\GG}{\tildex_i} = \outcomepair{b}{b,\bartildex_i,T(\tildex_i)}$  (if $\phi(\tildex_i) = 0$);
in both cases $\DealDv{\GG}{\tildex_i}  \succeq \DealDv{\DG}{\tildex_i}$.
For $v = \tildez_i$, similarly, either 
$\DealDv{\GG}{\tildez_i} = \outcomepair{b,\tildex_i}{b,T(\tildex_i)}$ (if $\phi(\tildex_i)= 1$)
or 
$\DealDv{\GG}{\tildez_i} = \outcomepair{b}{b}$ (if $\phi(\tildex_i)= 0$);
in both cases $\DealDv{\GG}{\tildez_i}  \succeq \DealDv{\DG}{\tildez_i}$.
Finally, we examine the vertices in the $\forall$-gadget. If $v \in \braced{p_g}_{g=0}^m \cup \braced{y_j,\bary_j}_{j=1}^l \cup \braced{q_j}_{j=0}^l$
then $\DealDv{\GG}{v} = \outcomepair{b}{b}\succeq \DealDv{\DG}{v}$.
Consider a vertex $v = \tau_g$, for some $g$, and let $\tildex_i$ be the $\bfx$-literal in $\tau_g$.
If $\phi(\tildex_i)=1$ then $\DealDv{\GG}{\tau_g} = \outcomepair{b,\tildez_i}{b}$,
and if $\phi(\tildex_i)=0$ then $\DealDv{\GG}{\tau_g} = \outcomepair{b,\tildex_i}{b}$.
In both cases, $\DealDv{\GG}{\tau_j}  \succeq \DealDv{\DG}{\tau_j}$.

\smallskip

It remains to verify condition~(c.3). Let $\HG$ be a subgraph of $\DG$, and suppose that
$\HG$ dominates $\GG$, that is $\DealDv{\HG}{v} \succeq \DealDv{\GG}{v}$ for all vertices $v$ in $\HG$.
We will show that is possible only if $\HG$ is either equal to $\GG$ or to one of the two
strongly connected components of $\GG$.

$\HG$ cannot contain any arcs from $a$ to literals $\tildex_i$, because then it would not
dominate $\GG$ at vertex $a$. There is also no subgraph consisting of $a$ and $a'$ that
strictly dominates $\GG$. We can thus assume that $\HG$ is a subgraph of $\DG' = \DG\setminus\braced{a,a'}$.
Let also $\GG' = \GG\setminus \braced{a,a'}$.
The rest of the argument is divided into two cases, depending on whether $\HG$ includes
vertex $b$ or not.

Suppose first that $\HG$ includes vertex $b$.
In this case, we claim that $\HG = \GG'$, and therefore $\HG$ does not strictly dominate $\GG$.
To show this, observe first that since $\DealDv{\HG}{b} \succeq \DealDv{\GG}{b}$, $\HG$ must contain all incoming arcs of $b$.
So $\HG$ must in fact contain all vertices of $\DG'$.
And each vertex $v\in \DG'\setminus\braced{b}$ does not have any outcome better than $\DealDv{\GG}{v}$ that does
not have arc $(b,v)$. Therefore $\HG$ must also contain all outgoing arcs of $b$.

The idea now is to show that for each vertex $v\in \DG'\setminus\braced{b}$, the outcome of $v$ in $\GG'$ is already
best possible among the outcomes that have incoming and outgoing arcs from $b$. A more formal argument actually
focuses on arcs rather than vertices, and involves two observations:
(i) For each arc $(u,v)\in \GG'$, vertex $v$ does not have any outcome that does not include incoming arc $(u,v)$
and is better than $\DealDv{\GG}{v}$.
(ii) For each arc $(u,v)\in \DG'\setminus\GG'$, vertex $u$ does not have any outcome that 
includes outgoing arc $(u,v)$ and is better than $\DealDv{\GG}{u}$.
These observations imply that $\DealDv{\HG}{v} \succeq \DealDv{\GG}{v}$ for all $v\in \DG'$, implying in turn that $\HG = \GG'$, as claimed.

Both observations~(i) and~(ii) can be established through routine although a bit tedious inspection of all arcs in $\DG'$.
(The process here is the same as in the $\NP$-hardness proof in Section~\ref{sec: simple np-hardness}.)

We start with the vertices in the $\exists$-gadgets. 
Consider some $\tildex_i$ and suppose $\phi(\tildex_i) = 1$ (symmetric for when $\phi(\tildex_i) = 0$). 
There is no outcome of $\tildex_i$ better than $\DealDv{\GG}{\tildex_i} = \outcomepair{b, \bartildex_i}{b,\tildez_i}$ that does not 
include the incoming arc $(\bartildex_i, \tildex_i)$. Also, there is no better outcome that includes arc $(\tildex_i,\tau_j)$, 
for each term $\tau_j\in T(\tildex_i)$. 
For ${\bartildex}_i, \DealDv{\GG}{\bartildex_i} = \outcomepair{b}{b,\tildex_i,T(\bartildex_i)}$.
There is no outcome of ${\bartildex}_i$ better than $\DealDv{\GG}{\bartildex_i}$ that
includes arc $({\bartildex}_i, {\bartildez}_i)$.
For a vertex $\tildez_i$ (still assuming that $\phi(\tildex_i) = 1$), 
$\DealDv{\GG}{\tildez_i} = \DealDv{\DG}{\tildez_i} = \outcomepair{b,\tildex_i}{b,T(\tildex_i)}$;
there is no better outcome that does not include $(\tildex_i, \tildez_i)$.
Lastly, there is no outcome of $\bartildez_i$ better than $\DealDv{\GG}{\bartildez_i} = \outcomepair{b}{b}$.

We move on to the vertices in the $\forall$-gadget.
For any vertex $v \in \smbraced{p_g}_{g=0}^m \cup \smbraced{\tildey_j}_{j=1}^l \cup \smbraced{q_j}_{j=0}^l$
we have $\DealDv{\GG}{v} = \outcomepair{b}{b}$ and, by the earlier argument, $\HG$ contains arc $(v,b)$.
But this $v$ does not have any outcome with outgoing arc $(v,b)$ that is better than $\DealDv{\GG}{v}$.
The argument when $v = \tau_g$, for some $g$, is similar. If the  unique $\bfx$-literal in $\tau_g$
is $\tildex_i$, then $\DealDv{\GG}{\tau_g} = \outcomepair{b, \tildez_i}{b}$ (if $\phi(\tildex_i) = 1$) or
$\DealDv{\GG}{\tau_g} = \outcomepair{b, \tildex_i}{b}$ (if $\phi(\tildex_i) = 0$).
In either case, as before, there is no outcome better than $\DealDv{\GG}{\tau_g}$ among the outcomes of $\tau_g$ that contain an outgoing arc to $b$.

\smallskip

Next, we consider the case when $\HG$ does not include vertex $b$. First, we observe that
$\HG$ cannot contain any vertices in the $\exists$-gadgets (namely vertices $\tildex_i$ and $\tildez_i$).
This is because for these vertices $v$ there is no outcome that is better than $\DealDv{\GG}{v}$ and does not 
include the incoming arc from $b$.

We can thus assume that $\HG$ is a subgraph of the $\forall$-gadget. (This is actually the most crucial case.)
Let $\DG''$ be the subgraph of $\DG$ induced by the vertices in the $\forall$-gadget.
Observe that every vertex $v$ in $\DG''$ has at least one outcome better than $\DealDv{\GG}{v}$ that does not include arcs to and from $b$,
so now we need a more subtle argument than the one we used earlier. 
For $v = \tau_g$, there are two cases.  The first is when $\tau_g$ has an incoming arc from its unique $\bfx$-literal $\tildex_i$ 
(which means $\phi(\tildex_i)=0$), in which case $\DealDv{\GG}{\tau_g} = \outcomepair{b, \tildex_i}{b}$.
By the preference poset of $\tau_g$, $\tau_g$ can improve this outcome by switching to $\outcomepair{p_{g-1},L}{p_g}$, 
for any set $L$ of the $\bfy$-literals in $\tau_g$. 
That is, this $\tau_g$ can improve its outcome regardless of whether it receives any arcs from its $\bfy$-literals.
The second case is when $\tau_g$ does not have an incoming arc from its $\bfx$-literal $\tildex_i$  
(which means $\phi(\tildex_i)=1$), in which case $\DealDv{\GG}{\tau_g} = \outcomepair{b, \tildez_i}{b}$.
By the preference poset of $\tau_g$, $\tau_g$ can improve its outcome by switching to $\outcomepair{p_{g-1},L'}{p_g}$ 
for any non-empty subset $L'$ of the $\bfy$-literals in $\tau_g$.
That is, this $\tau_g$ can improve its outcome only if it receives an arc from at least one of its $\bfy$-literals.
For $v = \tildey_j$, $\DealDv{\GG}{\tildey_j} = \outcomepair{b}{b}$. By the preference poset of $\tildey_j$, $\tildey_j$ 
can improve its outcome by switching to $\outcomepair{q_{j-1}}{q_{j},T(\tildey_j)}$, which results in creating arcs to the terms in $T(\tildey_j)$.
For $v = q_j$, $\DealDv{\GG}{q_j} = \outcomepair{b}{b}$. 
By the preference poset of $q_j$, where $j\notin\braced{0,l}$, the following outcomes of $q_j$ are better than $\DealDv{\GG}{q_j} \suchthat$
$\outcomepair{y_j}{y_{j+1}}$, $\outcomepair{\bary_j}{y_{j+1}}$, $\outcomepair{y_j}{\bary_{j+1}}$ or $\outcomepair{\bary_j}{\bary_{j+1}}$.
This means the preference posets of $q_{j-1}$ and $q_j$ allow only one of $y_j$ or $\bary_j$ to make the switch described above. 
(This corresponds to choosing which of these two literals is false.)
The same reasoning holds for $q_0$ and $q_l$, except their improved outcomes are $\outcomepair{p_m}{\tildey_1}$ and $\outcomepair{\tildey_l}{p_0}$ respectively.
For $v = p_g$, $\DealDv{\GG}{p_g} = \outcomepair{b}{b}$.
By the preference poset of $p_g$, where $g\notin\braced{0,m}$, $p_g$ can improve its outcome by switching to $\outcomepair{\tau_g}{\tau_{g+1}}$.
This means $p_g$ can only switch given that $\tau_g$ makes one of switches described above 
(either from $\outcomepair{b, \tildex_i}{b}$ to $\outcomepair{p_{g-1},L}{p_g}$ or from
$\outcomepair{b, \tildez_i}{b}$ to $\outcomepair{p_{g-1},L'}{p_g}$).
The same reasoning holds for $p_0$ and $p_m$, except their improved outcomes are $\outcomepair{q_l}{\tau_1}$ and $\outcomepair{\tau_m}{q_0}$ respectively.

Importantly, the outcome improvements in the above paragraph are possible only if \emph{all the vertices in $\DG''$ together
switch their outcomes as described in the above paragraph}.
This would correspond to choosing a subgraph $\HG$ that strictly dominates $\GG$ (namely the cycle given in~(\ref{eqn: sigma2-complete cycle})).
We now show this subgraph $\HG$ cannot exist, by way of contradiction.
Suppose such a subgraph $\HG$ that strictly dominates $\GG$ does exist.
Since $\HG$ strictly dominates $\GG$, and all vertices must improve together, we know every vertex $v \in \HG$ strictly improves their outcome from $\DealDv{\GG}{v}$.
We focus on the outcome improvements made by the term vertices $\tau_1...\tau_m$.
Let us fix some term vertex $\tau_g$ and let $\tildex_i$ be the unique $\bfx$-literal of $\tau_g$.

As described above, $\tau_g$ can improves its outcome in one of two ways, depending on $\DealDv{\GG}{\tau_g}$; specifically whether or not $(\tildex_i, \tau_g) \in \GG$.
If $(\tildex_i, \tau_g) \in \GG$, then $\tau_g$ can improve its outcome from $\DealDv{\GG}{\tau_g}$ 
by simply ``switching''.
Otherwise, if $(\tildex_i, \tau_g) \not\in \GG$, then $\tau_g$ can only switch to an improved outcome if it receives an arc from any of its $\bfy$-literals in $\HG$.  
In other words, each $\tau_g$ must have either received its incoming arc from its $\bfx$-literal in $\GG$ or received an incoming arc from any of its $\bfy$-literals in $\HG$.

Recall though that $\tau_g$ receives an arc from one of its literals only if that literal is set to false.
This implies that each term $\tau_g$ is killed, either by its $\bfx$-literal or one of its $\bfy$-literals, depending on how it improves its outcome.
However, if each term is killed under the assignments $\bfx \mapsto \bfphi$ and $\bfy \mapsto \bfpsi$, we know $\beta(\bfphi, \bfpsi)$ is false, contradicting our original assumption.

We show this more formally, starting with the terms being killed by the assignment of the $\bfx$ variables.
In graph $\GG$, for each variable $x_i$, if $\phi(x_i) = 1$, then for each term $\tau_g$ that contains $\barx_i$, $(\barx_i, \tau_g) \in \GG$.
On the other hand, if $\phi(x_i) = 0$, then for each term $\tau_g$ that contains $x_i$, $(x_i, \tau_g) \in \GG$.
In both cases, $\tau_g$ is killed.
Within the swap system, this is signified by vertex $\tau_g$'s preference to switch from $\DealDv{\GG}{\tau_g}$ to $\outcomepair{p_{g-1},L}{p_g}$.

Now we address the terms survived by the assignment $\bfx \mapsto \bfphi$.
The surviving term vertices are those that did not receive their incoming arcs from their $\bfx$-literals in $\GG$.
Since we know each surviving term vertex $\tau_g$ strictly improves their outcome in $\HG$, the only remaining option is that each $\tau_g$ has an incoming arc from one of their $\bfy$-literals in $\HG$.

We use this to construct the assignment $\bfy \mapsto \psi$ so that $\beta(\bfphi, \bfpsi)$ is false.
This is quite simple:
for each $\bfy$-literal $\tildey_j$ that has an outgoing arc to a surviving term vertex in $\HG$, we assign $\psi(\tildey_j) = 0$.
We know that $\bfpsi$ must be a consistent assignment, i.e. it cannot be the case that $\tildey_j$ and $\bartildey_j$ are both assigned to true/false.
This is because only either $\tildey_j$ or $\bartildey_j$ are in $\HG$, by design of the preference posets of vertices $q_{j-1}$ and $q_j$.
Thus, since we can construct a consistent assignment $\bfy \mapsto \bfpsi$, given the assignment $\bfx \mapsto \bfphi$, so that every term is killed, we know that $\beta(\bfphi, \bfpsi)$ is false, contradicting our original assumption.


\smallskip
$(\Leftarrow)$
Assume now that $\DG$ has a spanning subgraph $\GG$ that satisfies properties~(c.1) and~(c.2).
From $\GG$ we will construct an assignment $\bfphi$ for the $\bfx$-variables that makes $\forall \bfy \beta(\bfphi,\bfy)$ true.
Condition~(c.1) implies that $\GG$ cannot have any arcs $(a, \tildex_i)$ nor $(a, \tildez_i)$, so vertices $\braced{a,a'}$ will form one strongly connected component of $\GG$.
As before, let $\DG' = \DG\setminus \braced{a,a'}$ and $\GG' = \GG\setminus \braced{a,a'}$.
We focus on $\GG'$.

We first argue that $\GG'$ is in fact strongly connected and it contains $b$. This is quite simple. 
Condition~(c.2) states that the outcome of $b$ in $\GG$ is at least as good as its outcome in $\DG$, so $\GG'$ must contain all incoming arcs of $b$.
On the other hand, each vertex $v \in \GG'\setminus \braced{b}$ does not have an outcome better than $\DealDv{\DG}{v}$ that includes outgoing arc $(v,b)$
but does not include incoming arc $(b,v)$. Thus, $\GG'$ must also contain all outgoing arcs of $b$, 
which is already sufficient to make $\GG'$ strongly connected.


For each literal vertex $\tildex_i$,
we will refer to any outcome that contains $T(\tildex_i)$ in its set of outgoing arcs as a $0$-outcome of $\tildex_i$,
and to the exact outcome $\outcomepair{b,\bartildex_i}{b, \tildez_i}$ as the $1$-outcome $\tildex_i$.
We start with the following claim:

\smallskip

\myclaim{1}{
For each $i$ and each literal $\tildex_i\in\braced{x_i,\barx_i}$, outcome $\DealDv{\GG}{\tildex_i}$ is either a
$0$-outcome or the $1$-outcome of $\tildex_i$. 
Further, for at least one of $x_i$ and $\barx_i$ this outcome is a $0$-outcome. 
}

\smallskip

\begin{proof}
Let us fix a single $\exists$-gadget.
We first show that for literal $\tildex_i \in \braced{x_i, \barx_i}$, the outcome $\DealDv{\GG}{\tildex_i}$ is either a $0$-outcome or the $1$-outcome of $\tildex_i$.
Firstly, we know the incoming and outgoing arcs between $\tildex_i$ and vertex $b$ are included in $\GG'$.
Next, consider any term vertex $\tau_g$ in which term $\tau_g$ contains literal $\tildex_i$.
If we examine the generators of vertex $\tau_g$, limiting ourselves only to the outcomes that include the arcs to and from vertex $b$, we see that $\tau_g$ must receive either an arc from $\tildex_i$ or $\tildez_i$ in order to satisfy condition~(c.2).

We now have two cases: when $\tau_g$ receives an arc from $\tildex_i$ and when $\tau_g$ receives an arc from $\tildez_i$. 
We start with the latter case.
If $\tau_g$ receives arc $(\tildez_i, \tau_g)$, then by $\tildez_i$'s generators, we know that $\tildez_i$ must have received arc $(\tildex_i, \tildez_i)$.
This then implies that $\tildex_i$ received arc $(\bartildex_i, \tildex_i)$.
At this point, $\tildex_i$ is exactly in the $1$-outcome.
We reason similarly about $\bartildex_i$:
starting from some vertex $\tau_g$ for which term $\tau_g$ contains $\bartildex_i$, we know that $\tau_g$ must receive either an arc from $\bartildex_i$ or $\bartildez_i$.
We know $\tau_g$ cannot receive an arc from $\bartildez_i$ because for $\bartildez_i$ to pay arc $(\bartildez_i, \tau_g)$, it must receive arc $(\bartildex_i, \bartildez_i)$.
However, there is no outcome for $\bartildex_i$ that satisfies condition~(c.2) 
in which $\bartildex_i$ pays both arcs $(\bartildex_i, \tildex_i)$ and $(\bartildex_i, \bartildez_i)$. 
Thus, we can conclude that $\bartildex_i$ is the one to pay $\tau_g$.
We can reason about each $\tau_g \in T(\bartildex_i)$ in the same manner, implying that $\bartildex_i$ in fact pays every $\tau_g \in T(\bartildex_i)$.  
This allows us to conclude that $\bartildex_i$ is in a $0$-outcome.

We move on to the former case, when $\tau_g$ receives an arc from $\tildex_i$.
It is easy to see that if $\tildex_i$ pays any term vertex $\tau_g \in T(\tildex_i)$, it must pay all term vertices in $T(\tildex_i)$.
This is because each term vertex $\tau_g \in T(\tildex_i)$ must receive an arc from either $\tildex_i$ or $\tildez_i$, as previously stated.
However, there is no outcome for $\tildex_i$ that satisfies condition~(c.2)
in which $\tildex_i$ pays $\tau_g$ and $\tildez_i$, thus $\tildex_i$ is responsible for paying all term vertices $\tau_g \in T(\tildex_i)$.
This is sufficient to show that $\tildex_i$ is in a $0$-outcome.
We move onto vertex $\bartildex_i$.
Unlike the previous case, the outcome of $\bartildex_i$ is not directly influenced by the outcome of $\tildex_i$.
When we consider some term vertex $\tau_g \in T(\bartildex_i)$,
it is possible for $\tau_g$ to receive an arc from either $\bartildex_i$ or $\bartildez_i$.
We show that $\bartildex_i$ ends in a $0$-outcome or the $1$-outcome, respectively.
The first possibility is that $\tau_g$ receives arc $(\bartildex_i, \tau_g)$.
We apply the same reasoning as we did for $\tildex_i$:
if any $\tau_g \in T(\bartildex_i)$ receives its arc from $\bartildex_i$, then every $\tau_g \in T(\bartildex_i)$ also receives its arc from $\bartildex_i$.
This is again sufficient to show that $\bartildex_i$ is in a $0$-outcome.
The second possibility is that $\tau_g$ receives arc $(\bartildez_i, \tau_g)$.
For $\bartildez_i$ to pay this arc, it must receive arc $(\bartildex_i, \bartildez_i)$.
For $\bartildex_i$ to pay this arc, it must receive arc $(\tildex_i, \bartildex_i)$.
However, this is exactly the $1$-outcome for $\bartildex_i$.
We note that this requires $\tildex_i$ to pay arc $(\tildex_i, \bartildex_i)$, changing the outcome of $\tildex_i$.
Importantly though, $\tildex_i$ remains in a $0$-outcome and still satisfies condition~(c.2) as $\Dealv{\tildex_i} \prec \outcomepair{b}{b,{\bartildex}_i, T(\tildex_i)}$

It is easy to see that these two cases are exhaustive by inspection of the preference posets of $\tau_g$.
With this, we have shown both parts of claim~(1):
firstly, for each $i$, $\tildex_i$ and $\bartildex_i$ are either in a $0$-outcome or the $1$-outcome, and
secondly, at least one of $\tildex_i$ or $\bartildex_i$ are in a $0$-outcome, regardless of which case.
\end{proof}

For convenience, we now introduce the concept of a \emph{pseudo-truth assignment}. 
A pseudo-truth assignment is an assignment $\bfxi$ of boolean values to the $\bfx$-literals (not just variables) such that for each variable $x_i$ at most one of $\bfxi(x_i)$ and $\bfxi(\barx_i)$ is $1$.
The value of $\forall \bfy \beta(\bfxi,\bfy)$, for such a pseudo-truth assignment $\bfxi$, can be computed just like for standard truth assignments. 
If $\alpha$ has a satisfying pseudo-truth assignment $\bfxi$ then it also has a satisfying standard truth assignment $\bfphi$: simply let $\bfphi(x_i) = \bfxi(x_i)$ for all $i$. 
This works because if a term $\tau_g$ of $\beta$ is not killed by $\bfxi$ then it is also not killed by $\bfphi$.

Thus it suffices to show how we can convert $\GG$ into a pseudo-truth assignment $\bfxi$ for the $\bfx$-variables that
satisfies $\alpha$.
We define $\bfxi$ as follows: for each $i$, 
if $\DealDv{\GG}{\tildex_i}$ is a $0$-outcome then $\bfxi(\tildex_i) = 0$, and 
if $\DealDv{\GG}{\tildex_i}$ is the $1$-outcome then $\bfxi(\tildex_i) = 1$.

\smallskip
\myclaim{2}{
$\bfxi$ is a satisfying pseudo-truth assignment for the $\bfx$-variables that satisfies $\alpha$.
}
\smallskip

\begin{proof}
We begin by supposing the pseudo-truth assignment $\bfxi$ is not a satisfying assignment for $\alpha$, towards contradiction.
This would mean that $\forall y \beta(\bfxi, \bfy)$ is false.
We fix an assignment of the $\bfy$-variables $\bfpsi$ such that $\beta(\bfxi, \bfpsi)$ is false.
The idea is to now take $\bfpsi$ and construct a subgraph $\HG$ that strictly dominates $\GG$, contradicting our original assumption. 
Actually, $\HG$ will be a subgraph of the $\forall$-gadget of the form given in~(\ref{eqn: sigma2-complete cycle}), as before.

We now construct $\HG$ as follows:
add all vertices $v \in \smbraced{p_g}_{g=0}^m \cup \smbraced{q_j}_{j=0}^l \cup \smbraced{\tau_g}_{g=1}^m$ to $\HG$.
For each $j$, if $\bfpsi(y_j) = 1$, add $\bary_j$, otherwise, if $\bfpsi(y_j) = 0$, add $y_j$ (we include the literal that is false).
Now that we have all the vertices, we must define the arcs.
Again, $\HG$ will have the form of the cycle given in~(\ref{eqn: sigma2-complete cycle}).
For each $\tildey_j \in \HG$, add arcs $(q_{j-1}, \tildey_j)$ and $(\tildey_j, q_j)$.
Add arcs $(q_l, p_0)$ and $(p_m, q_0)$.
For each $\tau_g \in \HG$, add arcs $(p_{g-1}, \tau_g)$ and $(\tau_g, p_g)$.
Lastly, for each $\tildey_j \in \HG$, add arcs $(\tildey_j, \tau_g)$ for $\tau_g \in T(\tildey_j)$.

The next step is to show that $\HG$ indeed strictly dominates $\GG$.
It is easy to see that for vertices 
$v \in \smbraced{p_g}_{g=0}^m \cup \smbraced{\tildey_j}_{j=1}^l \cup \smbraced{q_j}_{j=0}^l$, 
$\DealDv{\GG}{v} \prec \DealDv{\HG}{v}$ holds by simple inspection of each vertex's preference poset.
Thus, we focus on the term vertices $\tau_1,..,\tau_m$.
For each term vertex $\tau_g$, outcome $\DealDv{\HG}{\tau_g}$ is an improvement in comparison to $\DealDv{\GG}{\tau_g}$ only if (at least) one of the two following conditions are satisfied:
(1) $\tau_g$ received its incoming arc from its $\bfx$-literal in $\GG$, or 
(2) $\tau_g$ receives an incoming arc from any of its $\bfy$-literals in $\HG$. 

We claim that one of these two conditions holds for every term $\tau_g$.
Suppose this is not true, towards contradiction, and there is a term vertex $\tau_g$ that does not satisfy either condition.
Specifically, $\tau_g$ does not receive its incoming arc from its $\bfx$-literal in $\GG$, nor
does $\tau_g$ receive any of its incoming arcs from any of its $\bfy$-literals in $\HG$.
If this were the case, then $\tau_g$ is actually true, contradicting the fact that $\beta(\bfxi, \bfpsi)$ is false.
Let $\tildex_i$ be the $\bfx$-literal of $\tau_g$.
If $(\tildex_i, \tau_g) \not\in \GG$, then $\DealDv{\GG}{\tildex_i}$ is actually the $1$-outcome for $\tildex_i$.
This implies that $\bfxi(\tildex_i) = 1$.
Since $\tau_g$ does not satisfy the second condition, we know it does not receive a single arc from any of its $\bfy$-literals.
However, recall how we used $\bfpsi$ to construct $\HG$; a $\bfy$-literal is added to $\HG$ only if that literal is \emph{false} in $\bfpsi$.
This means each of these $\bfy$-literals of $\tau_g$ are actually \emph{true} in the original assignment of $\bfpsi$.
This implies that the term $\tau_g$ is actually true, contradicting $\beta(\bfxi,\bfpsi)$ being false.

This contradiction gives us the fact that every term vertex $\tau_g$ indeed improves their outcome from $\DealDv{\GG}{\tau_g}$.
With this, we have proven every vertex $v \in \HG$ improves their outcome from $\DealDv{\GG}{v}$, meaning $\HG$ strictly dominates $\GG$.
However, the existence of such an $\HG$ contradicts our condition~(c.3), implying claim~(2), that the pseudo-truth assignment $\bfxi$ is indeed a satisfying assignment of the $\bfx$-variables for $\alpha$.

\end{proof}


With the truth assignment $\phi$ defined, we need to show that the non-existence of an $\HG$ that strictly dominates $\GG$
implies that the expression $\forall \bfy \beta(\bfphi,\bfy)$ is true.
For this, it's easier to show the contrapositive, namely if there existed some assignment $\bfpsi$ for the
$\bfy$-variables for which $\forall \bfy \beta(\bfphi,\bfpsi)$ is false, we could convert $\bfpsi$ into a
subgraph $\HG$ that strictly dominates $\GG$.

We simply employ the exact same argument we saw in the proof for claim~(2).
We convert the assignment $\bfpsi$ in the exact same manner: for each $y_j$, if $\bfpsi(y_j) = 1$, add $\bary_j$ to $\HG$, otherwise, if $\bfpsi(y_j) = 0$, add $y_j$.
The remainder of $\HG$ is constructed in the exact same way as previously described.
Likewise, the proof that $\HG$ indeed strictly dominates $\GG$ is the same.
Since this contradicts condition~(c.3), we know that the expression $\forall y \beta(\phi, y)$ is in fact true.

\end{proof}

%% file: sigma_two_figs.tex

\begin{figure*}[ht]
\centering
\includegraphics[width=4.75in]{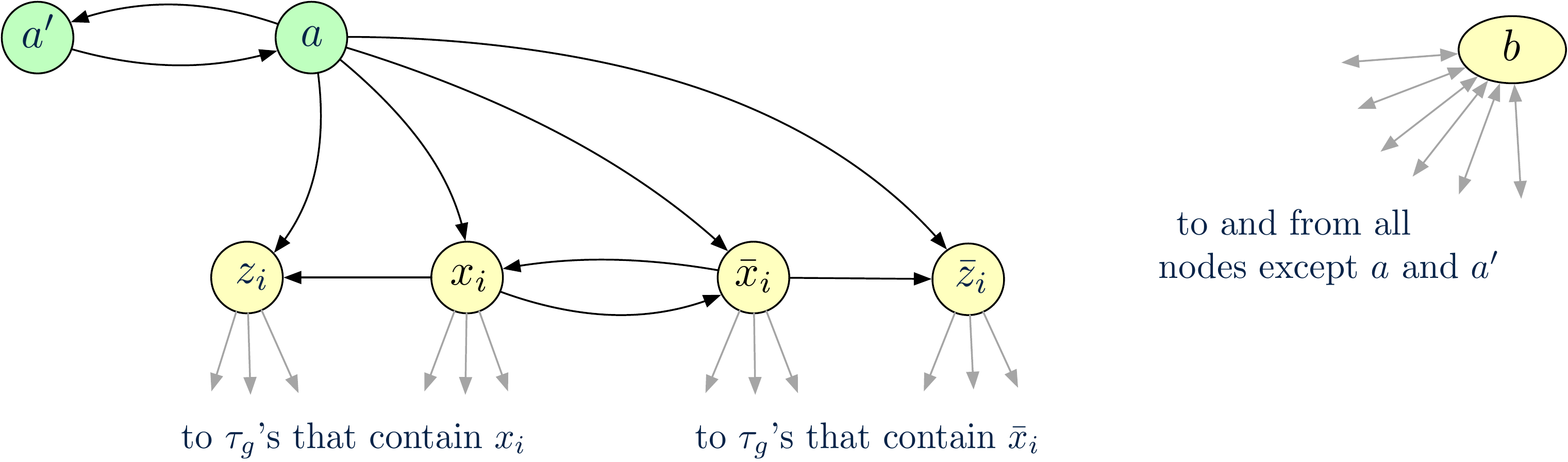} 
\caption{The construction of digraph $\DG$ in the proof of $\SigmaTwoP$-hardness. This figure shows vertices $a$, $a'$, $b$, and
an $\exists$-gadget for variable $x_i$.
The arcs to and from $b$ are shown as bi-directional arrows at $b$.}
\label{fig: sigma2 reduction 1}
\end{figure*}

\begin{figure*}[ht]
\centering
\includegraphics[width=6in]{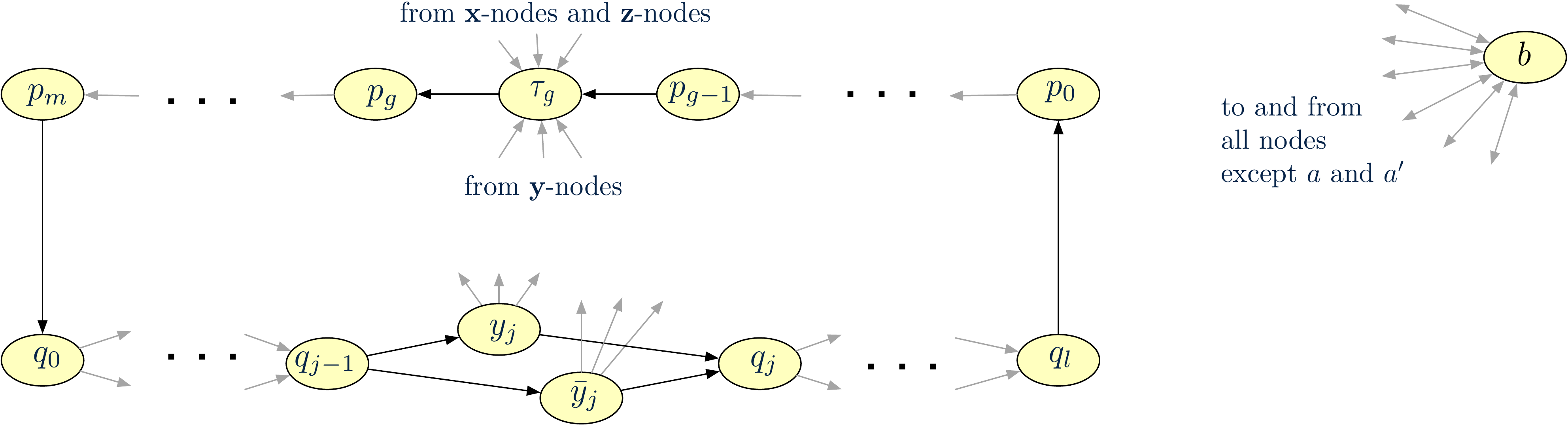} 
\caption{The construction of digraph $\DG$ in the proof of $\SigmaTwoP$-hardness. This figures shows the
$\forall$-gadget, namely the
part of $\DG$ that contains the vertices that simulate setting the values of the $y_j$-variables and the terms $\tau_g$.
The arcs to and from $b$ are shown as bi-directional arrows at $b$.}
\label{fig: sigma2 reduction 2}
\end{figure*}


%% file: 12_simplest_np_hardness.tex

In this section we give a proof of $\NP$-hardness of $\swapAtomic$ that is simpler than the one
in Section~\ref{sec:np_hardness_results}.

\begin{theorem}\label{thm: simple nphard}
$\swapAtomic$ is $\NP$-hard. It remains $\NP$-hard even for strongly connected digraphs.
\end{theorem}


\begin{proof}
The proof is by showing a polynomial-time reduction from $\CNF$. Recall that in $\CNF$
we are given a boolean expression $\alpha$ in conjunctive normal form, and
the objective is to determine whether there is a truth assignment that satisfies $\alpha$. 
In our reduction we convert $\alpha$ into a swap system $\swapSys = (\DG, \prefP)$
such that $\alpha$ is satisfiable if and only if  $\swapSys$ has an atomic swap protocol.

Let $x_1,x_2,...,x_n$ be the variables in $\alpha$.  The negation of $x_i$ is denoted $\barx_i$.
We will use notation $\tildex_i$ for an unspecified literal of variable $x_i$,
that is $\tildex_i \in \braced{x_i,\barx_i}$. Let $\alpha = c_1 \vee c_2 \vee ... \vee c_m$, where each $c_j$ is a clause.
Without loss of generality we assume that each literal appears in at least one clause and
that in each clause no two literals are equal or are negations of each other.

\smallskip

We first describe a reduction that uses a digraph $\DG$ that is not strongly connected. Later
we will show how to modify our construction to make $\DG$ strongly connected.
Digraph $\DG$ is constructed as follows (see Figure~\ref{fig: simple np-hardness gadgets}) :

{
\renewcommand\labelitemi{---}
\begin{itemize}[leftmargin=*]
\item 
For $i=1,...,n$, create vertices $x_i$ and $\barx_i$, connected by arcs $(x_i,\barx_i)$ and $(\barx_i,x_i)$. 
\item
Create two vertices $a,a'$ with arcs $(a,a')$, $(a',a)$, and $(a,x_i)$, $(a,\barx_i)$ for all $i = 1,...,n$.
\item
For $j = 1,...,m$, create vertices $c_j$. For each clause $c_j$ and each literal $\tildex_i$ in $c_j$,
create arc $(\tildex_i,c_j)$. 
\item
Create three vertices $d,d',d''$ with arcs $(d,d')$, $(d',d)$, $(d,d'')$, $(d'',d)$, $(d',d'')$ and $(d'',d')$.
Create also arcs $(c_j,d)$ for all $j = 1,...,m$.
\item 
Create vertex $b$, with arcs $(c_j,b)$ for all $j=1,...,m$ and $(b,x_i)$, $(b,\barx_i)$ for all $i = 1,...,n$.
\end{itemize}
}


\begin{figure*}
\centering
\includegraphics[width=4.75in]{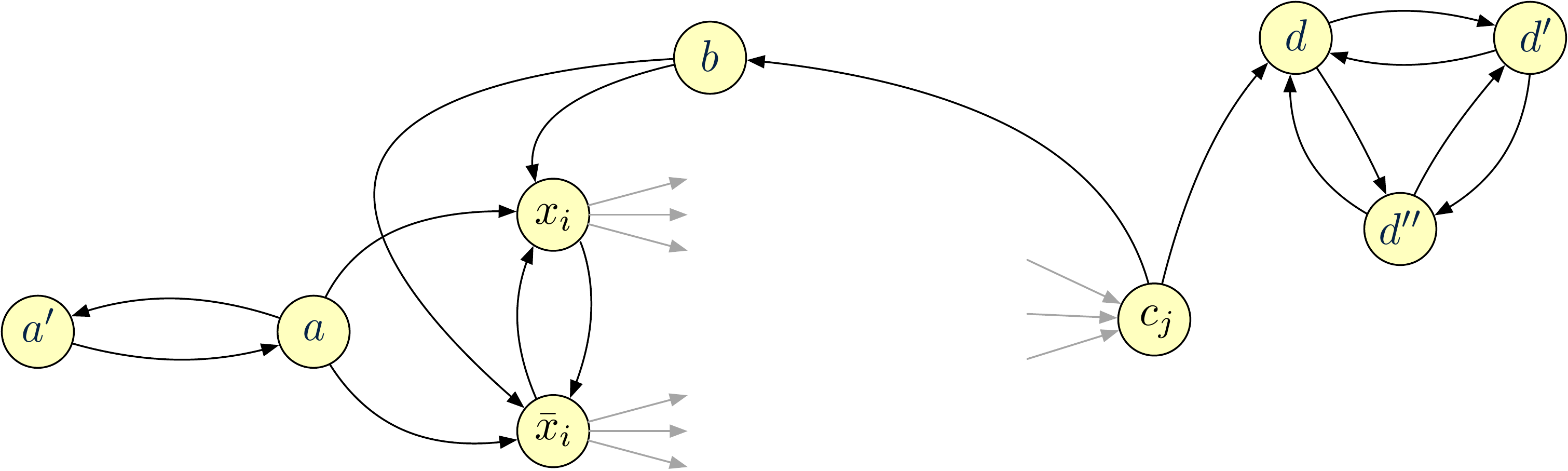} 
\caption{The variable and clause gadgets in the proof of Theorem~\ref{thm: simple nphard}.}
\label{fig: simple np-hardness gadgets}
\end{figure*}


Next, we describe the preference posets $\prefP_v$, for each vertex $v$ in $\DG$. 
As explained in Section~\ref{sec:swap systems}, an outcome $\outcomepair{\outcomein{}}{\outcomeout{}}$
of a vertex $v$ is specified by lists $\outcomein{}$ and $\outcomeout{}$ of its in-neighbors and out-neighbors.
The preference posets of the vertices in $\DG$ are specified by their generators:
{
\renewcommand\labelitemi{---}
\begin{itemize}[leftmargin=*]
\item 
Vertices $a$,$a'$, and $b$ do not have any generators. 
\item
For each literal $\tildex_i$, its generators are
$\Dealv{\tildex_i} \prec \outcomepair{b,\bar{\tildex}_i}{C(\tildex_i)}$ and $\Dealv{\tildex_i} \prec \outcomepair{b}{\bar{\tildex}_i}$,
where $\bar{\tildex}_i$ is the negation of $\tildex_i$ and $C(\tildex_i)$ is the set of clauses that contain literal $\tildex_i$.
\item
For each $j$, the generators of $c_j$ are
$\Dealv{c_j} \prec \outcomepair{\tildex_i}{b}$ for each literal $\tildex_i$ in $c_j$.
\item
Vertices $d,d',d''$ have one generator each:
$\Dealv{d} \prec \outcomepair{d''}{d'}$, $\Dealv{d'} \prec \outcomepair{d}{d''}$, $\Dealv{d''} \prec \outcomepair{d'}{d}$.
\end{itemize}
}

The construction of $\swapSys$ clearly takes time that is polynomial in the size of $\alpha$.


\smallskip

Applying Theorem~\ref{thm:uni-and-snash},  it remains to show that
$\alpha$ is satisfiable if and only if $\DG$ has a spanning subgraph $\GG$ with the following properties:
{(c.1)} $\GG$ is piece-wise strongly connected and has no isolated vertices,
{(c.2)} $\GG$ dominates $\DG$, and
{(c.3)} no subgraph $\HG$ of $\DG$ strictly dominates $\GG$.

\smallskip

$(\Rightarrow)$
Suppose that $\alpha$ is satisfiable, and fix some satisfying assignment for $\alpha$.
Using this assignment, we construct a spanning subgraph $\GG$ of $\DG$ that satisfies conditions (c.1)-(c.3).

Digraph $\GG$ will contain all vertices of $\DG$.
For vertices $a$ and $a'$ it will include arcs $(a,a')$ and $(a',a)$. For vertex $b$, it will include 
all arcs $(b,x_i)$, $(b,\barx_i)$ and all arcs $(c_j,b)$. 
Vertices $d,d',d''$ are connected by arcs $(d,d')$, $(d',d'')$ and $(d'',d)$.
The remaining arcs are determined based on the satisfying assignment.
Suppose that literal $\tildex_i$ is true. Then $\GG$ includes the  arcs: $(\bar{\tildex}_i,\tildex_i)$ and
$(\tildex_i,c_j)$ for all clauses $c_j$ that contain literal $\tildex_i$.
(Intuitively, the truth assignment corresponds to the direction of the arc between $x_i$ and $\barx_i$ in $\GG$.)

\smallskip

Digraph $\GG$ is spanning and has three strongly connected components: one is the cycle $a\to a'\to a$,
another one is the cycle $d\to d'\to d''\to d$, and
the third consists of all other vertices. This third component is indeed strongly connected because each
clause $c_j$ has a true literal, say $\tildex_i$, so its corresponding vertex has incoming edge
$(\tildex_i,c_j)$. We then have arcs from all vertices $c_j$ to $b$ and from $b$ to each pair $x_i$ and $\barx_i$.
For each $i$, among $x_i$ and $\barx_i$, the true literal $\tildex_i$ 
is connected to all clauses where it appears (and it must appear at least once, by our assumption), 
and its negation $\bar{\tildex}_i$ is connected to $\tildex_i$. So~(c.1) holds.

Condition~(c.2) can be verified by inspection, namely checking
that $\DealDv{\DG}{v}\preceq \DealDv{\GG}{v}$ holds for each vertex $v$.
For example, consider some variable $x_i$ and assume that $x_i$ is true (the case when $x_i$ is false  is symmetric).
Then $\DealDv{\GG}{x_i} = \outcomepair{b,\barx_i}{C(x_i)} \succ \DealDv{\DG}{x_i}$, and
$\DealDv{\GG}{\barx_i} = \outcomepair{b}{x_i}\succ \DealDv{\DG}{\barx_i}$.
Next, consider some clause $c_j$.
Since our truth assignment satisfies $c_j$, $c_j$ has some true literal $\tildex_i$. Then $\GG$ will have arc $(\tildex_i,c_j)$.
Denoting by $T(c_j)$ the set of true literals in $c_j$, we then have
$\DealDv{\GG}{c_j} = \outcomepair{T(c_j)}{b} \succeq \outcomepair{\tildex_i}{b}\succ \DealDv{\DG}{c_j}$.
Checking that $\DealDv{\DG}{v}\preceq \DealDv{\GG}{v}$ holds for $v\in\braced{a,a',b,d,d',d''}$
is straighforward. Thus, condition~(c.2) is verified.

To establish condition~(c.3), let $\HG$ be a subgraph of $\DG$ that dominates $\GG$,
that is $\DealDv{\HG}{v}\succeq \DealDv{\GG}{v}$ for each vertex $v\in\HG$. We claim that
then in fact we must have $\HG = \GG$, which will imply~(c.3). This claim follows from the following two observations:
(i) For each arc $(u,v)\in \GG$, vertex $v$ does not have any outcome that does not include incoming arc $(u,v)$
and is better than $\DealDv{\GG}{v}$.
(ii) For each arc $(u,v)\in \DG\setminus\GG$, vertex $u$ does not have any outcome that 
includes outgoing arc $(u,v)$ and is better than $\DealDv{\GG}{u}$.

These observations can be verified by inspection. Starting with $a$, for each literal $\tildex_i$,
there is no outcome of $a$ that is better than $\DealDv{\GG}{a}$ that includes arc $(a,\tildex_i)$ or does not include arc $(a',a)$.
For $a'$, there is no outcome better than $\DealDv{\GG}{a} = \outcomepair{a}{a}$ that does not include arc $(a,a')$.
Consider some $x_i$, and suppose that $x_i$ is true in our truth assignment. 
There is no outcome of $x_i$ better than $\DealDv{\GG}{x_i} = \outcomepair{b,\barx_i}{C(x_i)}$ that does not
include arcs $(b,x_i)$ and $(\barx_i,x_i)$, or that includes arc $(x_i,\barx_i)$. 
Regarding $\barx_i$, there is no outcome of $\barx_i$ better than $\DealDv{\GG}{\barx_i}$
that does not have arc $(b,\barx_i)$ or that has any arc $(\barx_i,c_j)$, for some clause $c_j$.
Next, consider arcs between literals and clauses. 
For a clause $c_j$ we have $\DealDv{\GG}{c_j} = \outcomepair{T(c_j)}{b}$.
There is no outcome of $c_j$ that misses one of the arcs from $T(c_j)$ or includes arc $(c_j,d)$
and is better than
$\outcomepair{T(c_j)}{b}$. (And we have already showed that in $\HG$, vertex $c_j$ cannot have arcs from its false literals.)
There is also no outcome of $b$ without arc $(c_j,b)$ better than $\DealDv{\GG}{b}$.
The verification of the two observations for the arcs between $d$, $d'$ and $d''$ can be carried out in the same  manner.


\smallskip

\noindent
$(\Leftarrow)$ 
Assume now that $\DG$ has a spanning subgraph $\GG$ that satisfies properties~(c.1) and~(c.2). (We will not use~(c.3) for now). 
From $\GG$ we will construct a satisfying assignment for $\alpha$.
Condition~(c.1) implies that $\GG$ cannot have any arcs $(a,\tildex_i)$, so vertices $a,a'$ will form one strongly
connected component of $\GG$.  Similarly, $\GG$ cannot have any arcs $(c_j,d)$, so vertices $d,d',d''$ will
also form a strongly connected component. In the rest of the argument we focus on the remaining vertices.

For each literal $\tildex_i$, since $\DealDv{\GG}{\tildex_i} \succeq \DealDv{\DG}{\tildex_i}$, and also using the preferences of $\tildex_i$,
we obtain that $\GG$ must have arc $(b,\tildex_i)$. Similarly, using the preferences of $b$, $\GG$ must contain all arcs $(c_j,b)$.
(This also follows from the fact that $c_j$'s cannot be singleton strongly connected components of $\GG$.)
This means that all vertices $b$, $\tildex_i$ and $c_j$ are in the same connected component of $\GG$ which,
by property~(c.1), must be strongly connected.

From the above paragraph, by strong connectivity, for each $i$ either $x_i$ or $\barx_i$ must have an arc to some clause vertex.
Also, since $\DealDv{\GG}{x_i} \succeq \DealDv{\DG}{x_i}$,
if $x_i$ has an arc to a clause vertex then $\GG$ must have arc $(\barx_i,x_i)$ and $\GG$ cannot have arc
$(x_i,\barx_i)$. In turn, since $\DealDv{\GG}{\barx_i} \succeq \DealDv{\DG}{\barx_i}$,
$\barx_i$ has no arcs in $\GG$ to any clause vertices.
Summarizing, we have this: exactly one of arcs $(x_i,\barx_i)$ or $(\barx_i,x_i)$ is in $\GG$,
and if $(\tildex_i,\bar{\tildex}_i)$ is in $\GG$ then $\tildex_i$ does not have any arcs to clause vertices.
This allows us to define a satisfying assignment, as follows. If $\GG$ has arc $(\barx_i,x_i)$, set $x_i$ to true,
and if $\GG$ has arc $(x_i,\barx_i)$, then set $x_i$ to false. 

Using condition~(c.1), in $\GG$ each vertex $c_j$ must have at least one incoming arc
from some literal $\tildex_i$ in $c_j$. By the previous paragraph, this literal is true in our truth assignment,
so it satisfies $c_j$. This establishes that all clauses are satisfied.

\medskip

To prove the second statement in the lemma, we modify our construction. Note that in the above proof we did not
use property~(c.3) in the $(\Leftarrow)$ implication. If $\DG$ is strongly connected, then it's itself a candidate for
$\GG$, so the modified construction will need to rely on property~(c.3) somehow.

This modification is in fact quite simple. Add arcs from all literal vertices $\tildex_i$ to $a$, and
set the preferences of $a$ so that it prefers to drop the arcs to and from these literal vertices to form a coalition with $a'$.
We apply the same trick to vertex $d$: it will have arcs going back to all $c_j$'s, but it will
be happy to drop these arcs, as well as the arc from $d''$, in exchange for dropping the arc to $d'$.
Then in the proof for implication $(\Leftarrow)$ we use condition~(c.3) to argue that the arcs
from $a$ to all $\tildex_i$'s will not be in $\GG$, for otherwise a subgraph $\DG$ consisting of $a,a'$
and the arcs between them would strictly dominate $\GG$. For the same reason, $\GG$ will not have
arcs from $d$ to any $c_j$.
\end{proof}

\medskip
\emph{Comment:} The $\NP$-hardness result in Theorem~\ref{thm: simple nphard} holds even if we
require that preference posets are specified by listing all preference pairs (including the
generic ones). This can be shown by modifying the construction so that all vertices in $\DG$
have constant degree, and thus all preference posets will 
have constant size. To this end, we can use a variant of $\CNF$ where each clause has three
literals and each variable appears at most three times. Then the only vertices of unbounded
degree will be $a$, $b$, and $d$. For $a$, its set of outgoing arcs
can be replaced by a chain of vertices each with one outgiong arc to one outneighbor of $a$.
The same trick applies to the arcs of $b$ and $d$.

%% file: 13_supplemental_experiments.tex
To further study the 
complexity
of $\swapAtomic$ 
(\emph{i.e.}, given a swap system $\swapSys = (\DG, \prefP)$, decide whether it has an atomic protocol),
we programmed a simple implementation in C\texttt{++}.
We note that this algorithm would be run by the party assembling the swap system,
preceding any interaction with any blockchain.
This would normally be a market clearing service.

The algorithm runs in three phases.
Each phase is a filter for a condition in Theorem~\ref{thm:uni-and-snash}.
We start with every possible graph $\GG$, and pass each of them through the three filters.
If there is a graph remaining, then we decide yes, otherwise we decide no.
The first condition is that $\GG$ is spanning, piece-wise strongly connected, and contains no isolated vertices.
We first check that $\GG$ contains every vertex, each with at least one incoming and outgoing arc.
If so, we find the strongly connected components of $\GG$ using Kosaraju's algorithm \cite{sharir1981scc}.
We then check for every arc $(u,v)$ in $\GG$ that $u$ and $v$ are in the same component.
If so, then $\GG$ is piece-wise strongly connected, and we pass this graph to the second phase.

The second condition is that $\GG$ dominates $\DG$, the original digraph.
That is, for every vertex $v$, 
$\DealDv{\DG}{v} \preceq \DealDv{\GG}{v}$, 
where $\DealDv{\DG}{v}$ is the outcome for $v$ if every arc in $\DG$ were triggered, 
and $\DealDv{\GG}{v}$ is the outcome for $v$ if every arc in $\GG$ were triggered.
This is simple.
We say $\DealDv{\DG}{v} \succeq \DealDv{\GG}{v}$ if 
(1) they are the same outcome, 
(2) $\DealDv{\GG}{v}$ is inclusively monotone of $\DealDv{\DG}{v}$, or 
(3) $\DealDv{\DG}{v} \succ \DealDv{\GG}{v}$ by a non-generic generator (and transitivity).
If this holds for every vertex, then $\GG$ dominates $\DG$ and we pass $\GG$ to the third phase.

The last condition is that there is no subgraph $\HG$ of $\DG$ that strictly dominates $\GG$.
To verify this, we generate every possible subgraph $\HG$.
Then, for every vertex $v$ in $\HG$, we see if $\DealDv{\GG}{v} \preceq \DealDv{\HG}{v}$ and at least one vertex
where $\DealDv{\GG}{v} \prec \DealDv{\HG}{v}$.
If not, then $\HG$ does not strictly dominate $\GG$.
If no $\HG$ strictly dominates $\GG$, then we decide yes.
However, if after all three phases no graph remains, we decide no.

\myparagraph{Results and Assessment}
We ran this program on the example swap systems presented in this paper.
The program was written in C\texttt{++}11 and compiled with g\texttt{++} 12.2.0.
It was ran on a Windows 10 machine with a Intel Core i5-11400F 6-Core 2.6GHz CPU and 16 GB RAM.
We list the mean of ten runs of each swap system.
We provide three additional datapoints:
(1) number of arcs in the digraph,
(2) number of non-generic preferences generators, and
(3) whether or not the swap system ended up permitting an atomic protocol.
\\\

\noindent
\!\!\!\!\!
\begin{tabular}{ |l|l|l|l|l|  }
    \hline
    \multicolumn{5}{|c|}{Results} \\
    \hline
    
    Swap System  & Runtime & Arcs & Preferences & Protocol? \\
    \hline
    $\swapSys_1$ & 0.0567s   & 6     & 5     & Yes \\
    $\swapSys_2$ & 0.016s    & 6     & 2     & No  \\
    $\swapSys_3$ & 123.116s  & 14    & 14    & Yes \\
    $\swapSys_4$ & 61.851s   & 14    & 12    & No  \\
    $\swapSys_5$ & 328.904s  & 17    & 14    & No  \\
    \hline
\end{tabular}
\ \\

Swap system $\swapSys_1$ is the system defined in Example 1. 
Swap system $\swapSys_2$ is the system defined in Figure~\ref{fig:herlihy not strong nash}.
Swap system $\swapSys_3$ is the system defined in Example 4.
Swap system $\swapSys_4$ is $\swapSys_3$, except the two preference generators 
$\Dealv{t_1} \prec \outcomepair{t_{2}}{t_{2}}$ 
and 
$\Dealv{t_2} \prec \outcomepair{t_{1}}{t_{1}}$ 
are removed.
Swap system $\swapSys_5$ is $\swapSys_3$, except we add a new party $s_1$ and arcs $(u_1, s_1)$, $(u_2, s_1)$, and $(s_1, t_1)$.
Non-generic preferences are not changed.

As we can see in $\swapSys_1$ and $\swapSys_2$, it is feasible to compute $\swapAtomic$ for small swap systems, as expected.
The runtimes are less than a second.
We next look at larger graphs and highlight the difficulty of $\swapAtomic$.
We observe that $\swapSys_3$ and $\swapSys_4$ have higher runtimes.
Further, their runtimes are not in the same ballpark although
they have the same number of arcs.
Firstly, because piece-wise strong connectivity is a requirement, 
one might suspect that the cause is the number of arcs or the degree of the vertices.
However, the digraphs for both swap systems are exactly the same.
%
The natural reaction is to look at the preference posets.
We removed two generators from $\swapSys_3$ to $\swapSys_4$.
This made it so the swap system no longer had an atomic protocol, which reduced the runtime.
This is because in phase three, the program halts as soon as it finds an $\HG$ for every $\GG$ (that passed phases one and two).
On the other hand, when the system does permit a protocol,
the entirety of phase three needs to finish.
That is, it needs to check all possible $\HG$ to verify $\GG$ has no strictly dominating subgraphs.
Lastly, from $\swapSys_3$ to $\swapSys_5$, we added one party and three arcs, but no non-generic preferences were changed.
Although $\swapSys_5$ ended up not permitting a protocol, it scaled poorly with respect to $\swapSys_3$.

In practice, 
the runtimes may not be predictable, as is the case with $\NP$-Hard problems.
Needless to say, an increase in the number of arcs will generally increase the running time as there are more rounds of Kosaraju's algorithm in phase one.
Additionally, if one is to believe the swap system does indeed permit a protocol, then one should expect a long runtime as well, as the program needs to verify every subgraph in phase three.